\documentclass[11pt,titlepage]{amsart}

\usepackage{amsaddr}

\usepackage{hyperref}
\usepackage{cleveref}
\usepackage{latexsym,amsmath, bm}
\usepackage{enumerate}
\usepackage{amsfonts}
\usepackage{amssymb}
\usepackage{geometry}
\usepackage{latexsym}
\usepackage{fixmath}
\usepackage{faktor}
\usepackage{mathtools}
\usepackage[T1]{fontenc}
\usepackage{fourier}
\usepackage{bbm}
\usepackage{color}
\usepackage{pgfplots}
\usepackage{tikz}
\usetikzlibrary{shapes,decorations,arrows,calc,arrows.meta,fit,positioning}
\usepackage{graphicx}
\usepackage{fullpage}
\usetikzlibrary{decorations.pathreplacing, matrix}
\usepackage[boxed, linesnumbered, noend, noline]{algorithm2e}
\usepackage{float}
\usepackage{subcaption}

\numberwithin{equation}{section}

\newcommand{\vm}{\vec m}

\newcommand{\vN}{\vec N}

\newcommand{\vP}{\vec P}

\newcommand{\vX}{\vec X}
\newcommand{\vY}{\vec Y}

\newcommand\cE{\mathcal{E}}
\newcommand\cF{\mathcal{F}}

\newcommand\cP{\mathcal{P}}

\newcommand\SIGMA{\vec\sigma}

\newcommand{\Bin}{{\rm Bin}}

\newcommand{\Be}{{\rm Be}}

\renewcommand{\vec}[1]{\boldsymbol{#1}}
\newcommand{\vecone}{\vec{1}}

\newcommand\KL[2]{D_{\mathrm{KL}}\bc{{{#1}\|{#2}}}}

\newcommand\eps{\varepsilon}

\newcommand\pr{\mathbb{P}} 
\renewcommand\Pr{\pr}

\newcommand\Erw{\mathbb{E}}

\newcommand{\whp}{w.h.p.}

\newcommand\geb[1]{\textcolor{black}{#1}}
\newcommand\gebA[1]{\textcolor{black}{#1}}

\newcommand\change[1]{\textcolor{black}{#1}}
\newtheorem{definition}{Definition}[section]

\newtheorem{remark}[definition]{Remark}
\newtheorem{theorem}[definition]{Theorem}
\newtheorem{lemma}[definition]{Lemma}
\newtheorem{proposition}[definition]{Proposition}
\newtheorem{corollary}[definition]{Corollary}

\newcommand\Lem{Lemma}
\newcommand\Prop{Proposition}
\newcommand\Thm{Theorem}

\newcommand\Cor{Corollary}
\newcommand\Sec{Section}

\newcommand\bc[1]{\left({#1}\right)}
\newcommand\cbc[1]{\left\{{#1}\right\}}

\newcommand\brk[1]{\left\lbrack{#1}\right\rbrack}

\newcommand\abs[1]{\left|{#1}\right|}

\newcommand{\mzero}{\vm_{0}}
\newcommand{\mzerond}{\vm_{0,\text{nd}}}
\newcommand{\mone}{\vm_{1}}
\newcommand{\zeroplus}{V_{0,\text{PD}}}
\newcommand{\vzeroplus}{\vec{ V_{0,\textbf{PD}}}}
\newcommand{\mcomp}{m_{\text{COMP}}}
\newcommand{\mdd}{m_{\text{DD}}}
\newcommand{\mcount}{m_{\text{COUNT}}}
\newcommand{\Cchan}{C_{\text{Chan}}}
\newcommand{\G}{\vec G}
\newcommand{\VS}{\vspace*{10pt}}
\newcommand{\wt}[1]{{\widetilde{#1}}}
\newcommand{\dch}{d^*_{{\rm ch}}}

\begin{document}
	\title{Improved bounds for noisy group testing with constant tests per item}
	
	\author{Oliver Gebhard, Oliver Johnson, Philipp Loick, Maurice Rolvien}

    \address{ {\tt \{Gebhard, Loick, Rolvien\}@math.uni-frankfurt.de}, Goethe University, Mathematics Institute,\\ 10 Robert Mayer St, Frankfurt 60325, Germany.}
    \address{{\tt O.Johnson@bristol.ac.uk}, University of Bristol, School of Mathematics, \\ Woodland Road, Bristol, BS8 1UG, United Kingdom}
    
	\begin{abstract}
	\noindent
	The group testing problem is concerned with identifying a small set of infected individuals in a large population. At our disposal is a testing procedure that allows us to test several individuals together. In an idealized setting, a test is positive if and only if at least one infected individual is included and negative otherwise. Significant progress was made in recent years towards understanding the information-theoretic and algorithmic properties in this noiseless setting. In this paper, we consider a noisy variant of group testing where test results are flipped with certain probability, including the realistic scenario where sensitivity and specificity can take arbitrary values. Using a test design where each individual is assigned to a fixed number of tests, we derive explicit algorithmic bounds for two commonly considered inference algorithms and thereby \gebA{naturally extend the} results of Scarlett \& Cevher (2016) and  Scarlett \& Johnson  (2020). \gebA{We  provide improved performance guarantees for the efficient algorithms in these noisy group testing models -- indeed, for a large set of parameter choices the bounds provided in the paper are the strongest currently proved.}
	\end{abstract}
	\maketitle
	\newpage
\section{Introduction}

\subsection{Motivation and background}

Suppose we have a large collection of $n$ people, a small number $k$ of whom are infected by some disease, and where only $m \ll n$ tests are available. 
In a landmark paper \cite{Dorfman_1943} from 1943, Dorfman introduced the idea of group testing. The basic idea is as follows: rather than screen one person using one test, we could mix samples from individuals in one pool, and use a single test for this whole pool. The task is to recover the  infection status of all individuals using the pooled test results.
Dorfman's original work was motivated by a biological application, namely identifying individuals with syphilis. Subsequently, group testing has found a number of related applications, including detection of HIV \cite{Wein_1996}, DNA sequencing \cite{Kwang_2006,Ngo_2000} and protein interaction experiments \cite{Mourad_2013, Thierry_2006}. More recently, it has been recognised as an essential tool to moderate pandemic spread \cite{Cheong_2020}, where identifiying infected individuals fast and at a low cost is indispensable \cite{Madhav_2017}. In particular, group testing has been identified as a testing scheme for the detection of COVID-19 \cite{Abdalhamid_2020,Goethe_2020,Technion_2020}.
From a mathematical perspective, group testing is a prime example of an inference problem where one wants to learn a ground truth from (possibly noisy) measurements \cite{Abbe_2016,Arikan_2009,Donoho_2006}. 
Over the last decade, it has regained popularity and 
a significant body of research was dedicated to understand its information-theoretic and algorithmic properties \cite{ Baldassini_2013, Coja_2019, Coja_2019_2,Scarlett_2019, Scarlett_2016_2, Scarlett_2016}.
In this paper, we provide improved upper bounds on the number of tests that guarantee successful inference for the noisy variant of group testing.

\subsection{Related Work}
\subsubsection{Noiseless Group Testing}
In the simplest version of group testing, we suppose that a test is positive if and only if the pool contains at least one infected individual. We refer to this as the noiseless case. In this setting, each negative test guarantees that every member of the corresponding pool is not infected, so they can be removed from further consideration. However, a positive test only tells us that at least one item in the test is defective (but not which one), and so requires further investigation. 
Dorfman's original work \cite{Dorfman_1943}  proposed a simple adaptive strategy where a small pool of individuals is tested, and where each positive test is followed up by testing every individual in the corresponding pool individually. Since then it has been an important problem to find the optimal way to recover the whole population's infection status in the noiseless case \geb{(see \cite{Aldridge_2019_2} for a detailed survey)}. A simple counting argument (see for example \cite[Section 1.4]{Aldridge_2019_2}) shows that to ensure recovery with zero error probability, since every possible defective set must give different test outcomes, the following must hold  in the noiseless setting:
\begin{align}\label{count}
   2^m\geq \binom{n}{k} \qquad \Rightarrow \qquad m \geq m^{0}_{\inf} := \frac{1}{\log2} k \log(n/k)
\end{align}

\geb{This can be extended to the case of recovery with small error probability, for example with the bound (see \cite[Eq. (1.7)]{Aldridge_2019_2}) that the success probability 
\begin{align} \label{count2}
\Pr( \rm{ suc}) \leq \frac{2^m}{\binom{n}{k}},
\end{align}
meaning that the success probability must decay exponentially with the number of tests below $m^0_{\inf}$.}
Hwang \cite{Hwang_1972} provided an algorithm based on repeated binary search, which is essentially optimal in terms of the number of tests required in that it requires $m^0_{\inf} + O(k)$ tests, but may require many stages of testing.
The question of whether non-adaptive algorithms (or even adaptive algorithms with a limited number of stages) can attain the bound \eqref{count} remained open until recently. \cite{Aldridge_2018,Coja_2019_2} showed that the answer depends on the prevalence of the disease, for example on the value of $\theta \in (0,1)$ in a parameterisation\footnote{\gebA{The result of \cite{Coja_2019_2} is two-fold. On the one hand, it provides a method to recover infected individuals \whp as well as attaining \eqref{count} for a certain range of $\theta<\theta^*$. On the other hand they show that \eqref{count} cannot be attained by any testing procedure for larger $\theta>\theta^*$. One finds $\theta^*=\log(2)\cdot (1+\log(2))^{-1}$.}} where the number of infected individuals $k\sim n^{\theta}$. Non-adaptive testing schemes can be represented through a binary $(m \times n)$-matrix that indicates which individual participates in which test. Significant research was dedicated to see which design attains the optimal performance, 
\geb{although much of the recent research  analysed the performance of randomized designs.}
Initial research focused on the case where the matrix entries are i.i.d. \cite{Aldridge_2017,Aldridge_2014,Scarlett_2016}, which we will refer to as Bernoulli pooling. Later work considered a constant column design where each individual is assigned to a (near-)constant number of tests \cite{Aldridge_2016,Coja_2019,Coja_2019_2,Johnson_2019}. Indeed \cite{Coja_2019_2} showed that such a design is information-theoretically optimal in the \textit{noiseless} setting and it is to be expected that this remains true for the noisy case.
To recover the ground truth from the test results and the pooling scheme, this paper focuses on two non-adaptive algorithms, {\tt COMP} and {\tt DD}, which are relatively simple to perform and interpret in the noiseless case. We describe them in more detail below, but in brief {\tt COMP} \cite{Chan_2011} simply builds a list of all the individuals who ever appear in a negative test and are hence certainly healthy, and assumes that the other individuals are infected.
{\tt DD} \cite{Aldridge_2014} uses {\tt COMP} as a first stage and builds on it by looking for individuals who appear in a positive test that only otherwise contains individuals known to be healthy. 
While the noiseless case provides an interesting mathematical abstraction, it is clear that it may not be realistic in practice \cite{Plebani_2015}.
\subsubsection{Noisy Group Testing}
 In medical applications \cite{Ting_2011} the two occurring types of noise  in a testing procedure are related to sensitivity (\geb{the probability that a test containing an infected individual is indeed positive}) and specificity (\geb{the probability that a test with only healthy individuals is indeed negative}), and in that language we cannot assume the gold standard of tests with unit specificity and sensitivity.
Thus, research attention in recent years has shifted towards the noisy version of group testing \cite{Chan_2011,Scarlett_2018,Scarlett_2019,Scarlett_2016,Scarlett_2017,Johnson_2018}. On the one hand, the \textit{adaptive} noisy case was considered in \cite{Scarlett_2018,Scarlett_2019}. On the other hand \cite{Chan_2011,Johnson_2010,Knill_1996,Malioutov_2012,Scarlett_2016,Scarlett_2017,Johnson_2018} looked at the \textit{non-adaptive} noise case from different angles (for instance linear programming, belief propagation, and Markov Chain Monte Carlo). \gebA{In \cite{Scarlett_2016,Scarlett_2017, Johnson_2018} the  algorithmic performance guarantees within noisy group testing under Bernoulli pooling are discussed. First of all \cite{Scarlett_2016} obtained a converse as well as a theoretical achievability bound, but stated the practical recovery as an direction for further research. In the following \cite{Scarlett_2017, Johnson_2018} shed light on this question by using  Bernoulli pooling.}\footnote{\gebA{\cite{Scarlett_2017} introduced an approach based on separate decoding of items for symmetric noise models. While this approach works well for small $\theta$ (in particular $\theta\rightarrow 0$), the performance drops dramatically for larger $\theta$.  For most $\theta$ this approach is worse off than the noisy {\tt DD} discussed in \cite{Johnson_2018}. Note there exist some noise levels with the very strong restriction assuming $p=q$ where \cite{Scarlett_2017} improve over our results in the $\theta$ very close to 0 regime. Due to the generality of our model we will from now on focus on \cite{Johnson_2018} as benchmark for our results.}}
In this paper we focus on the {\tt COMP} and {\tt DD} algorithms, since it is possible to deduce explicit performance guarantees for them. The original {\tt COMP} and {\tt DD} were designed for the noiseless case and do not automatically carry over to general noisy models.
However, recent work of Scarlett and Johnson \cite{Johnson_2018} showed that noisy versions of these algorithms can perform well under certain noise models using  i.i.d. (Bernoulli pooling) test designs, particularly focusing on $Z$ channel and reverse $Z$ channel noise.
As common medical tests have different values for sensitivity and specificity \cite{Long_2018} the analysis of a generalized noise model beyond the $Z$ and reverse $Z$ channel is warranted.
\subsubsection{Model Justification}
As described for example in pandemic plans developed by the EU, US and WHO \cite{EU_2009,US_2017,WHO_2009}, and in COVID-specific work \cite{mutesa2020}, adaptive strategies may not be suitable for pandemic prevention. For example, if a test takes one day to prepare and for the results to be known, then each stage will require an extra day to perform, meaning that adaptive group testing information can be received too late to be useful. Hence the need to perform large-scale testing to identify infected individuals fast relative to the doubling time \cite{Cheong_2020, Madhav_2017, mutesa2020} can make adaptive group testing unsuitable to prevent an infectious disease from spreading.  Furthermore it may be difficult to preserve virus samples in a usable state for long enough to perform multi-round testing \cite{Gould_1999}. Due to its automation potential and the fact that tests can be completed in parallel (for example by the use of 96-well PCR plates \cite{erlich_2015}), the main applications of group testing such as DNA screening \cite{Chen_2008,Kwang_2006,Ngo_2000}, HIV testing \cite{Wein_1996} and protein interaction analysis \cite{Mourad_2013,Thierry_2006} are non-adaptive, where all tests are specified upfront and performed in parallel. For example, while group testing strategies appear to be useful to identify individuals infected with COVID-19 (see for example \cite{Goethe_2020,Technion_2020}), testing for the presence of the SARS-CoV-19 virus is not perfect \cite{Woloshin_2020}, and so we need to understand the effect of both false positive and false negative errors in this context, with non-identical error probabilities. For this reason, we consider a general $p-q$ noise model in this paper. Under this model, a truly negative test is flipped with probability $p$ to display a positive test result, while a truly positive test is flipped to negative with probability $q$ (Figure \ref{pqnoise}). Its formulation is sufficiently general to accommodate the recovery of the noiseless results ($p=q=0$), Z channel ($p=0$), reverse Z channel ($q=0$) and the Binary Symmetric Channel ($p=q$). However, our results include the case of non-zero $p$ and $q$ without having to make the somewhat artificial assumption that false negative and false positive errors are equally likely.  We note that it may be unrealistic to assume that the noise parameters are known exactly, and more sophisticated models may be needed to understand the real world. Nevertheless our analysis of a generalised noise model serves as a starting point towards a full understanding of the  difficulties occurring while implementing group testing algorithms in laboratories.

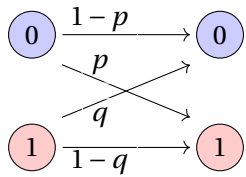
\begin{figure}[!ht]
\centering
\begin{tikzpicture}
\node (A) at (0,0) [circle,draw,fill=blue!20] {$0$};
\node (B) at (2.5,0) [circle,draw,fill=blue!20] {$0$};
\node (C) at (2.5,-1.5) [circle,draw,fill=red!20] {$1$};
\node (D) at (0,-1.5) [circle,draw,fill=red!20] {$1$};
\draw[->] (0.4,0)--(2.1,0);
\draw[->] (0.4,-1.5)--(2.1,-1.5);
\draw[->] (0.4,-0.4)--(2.1,-1.1);
\draw[->] (0.4,-1.1)--(2.1,-0.4);
\node (A) at (0.9,-0.4) [draw=white] {$p$};
\node (B) at (0.9,0.2) [draw=white] {$1-p$};
\node (C) at (0.9,-1.1) [draw=white] {$q$};
\node (D) at (0.9,-1.7) [draw=white] {$1-q$};
\end{tikzpicture}
\caption{The $p-q$-noise model: the result of each standard noiseless group test is transmitted independently through the given noisy communication channel.}
\label{pqnoise}
\end{figure}

\subsection{Contribution}
This paper provides a simultaneous extension of \cite{Coja_2019} and \cite{Johnson_2019,Johnson_2018}, by analysing noisy versions of {\tt COMP} and {\tt DD} under more general noise models for constant-column weight designs. \gebA{In contrast to prior work \cite{Aldridge_2014, Johnson_2019} assuming sampling with replacement, in this paper we use sampling without replacement, meaning that our designs have exactly the same number of tests for each item, rather than approximately the same as in those previous works. This makes little difference in practice, but may be closer to the spirit of LDPC codes for example.}

We provide explicit bounds on the performance of these algorithms in a generalized noise model. \gebA{We will prove that (noisy versions of) {\tt COMP} as well as {\tt DD} succeed with $\Theta(k\log(n/k))$ tests. Our analysis reveals the exact constants to ensure the recovery with these two inference algorithms. The main results will be stated formally in Theorems~\ref{thm_COMP} and~\ref{thm_DD}}, but we would like to give the reader a first insight of what will follow.
We analyze Algorithms \ref{comp_algorithm} and \ref{dd_algorithm} for the constant degree model, where there are $m=c k \log(n/k)$ tests performed and each individual chooses $\Delta = c d \log(n/k)$ tests uniformly at random. Let $p,q\geq 0, p+q<1$ and $\epsilon>0$.

We start with the performance of COMP (Algorithm~\ref{comp_algorithm}), as stated in Theorem \ref{thm_COMP}:

\textit{
For any $\Delta:=\Delta(c,d)$ we find a threshold $\alpha:=\alpha(d,p,q)$ such that COMP succeeds in inferring the infected  individuals if the number of tests $$m\geq (1+\eps)m_{COMP}=\min_{\alpha, d} \max \cbc{b_1(\alpha, d), b_2(\alpha, d)} k \log(n/k)$$ }

The next step on our agenda is the performance of DD (Algorithm~\ref{dd_algorithm}), as stated in Theorem \ref{thm_COMP}:

\textit{
For any $\Delta:=\Delta(c,d)$ we find  thresholds $\alpha:=\alpha(d,p,q)$ and $\beta:=\beta(d,q)$ such that DD succeeds in inferring the infected individuals if the number of tests $$m\geq (1+\eps)\mdd(n,\theta, p, q) = \min_{\alpha, \beta, d} \max \cbc{c_1(\alpha, d), c_2(\alpha, d), c_3(\beta, d), c_4(\alpha, \beta, d)} k \log(n/k)$$
}

For all typical noise channels (Z, reverse Z and BSC) we compare the constant-column and Bernoulli design and find for all such instances that the \gebA{required number of tests in the }former \gebA{is lower than the number needed in} the latter thereby improving on results from \cite{Johnson_2018}, and providing  the strongest performance guarantees currently proved for efficient algorithms in noisy group testing.

As group testing offers an essential tool for pandemic prevention \cite{Madhav_2017} and as the the accuracy of medical testing is limited \cite{Long_2018,Plebani_2015} this paper provides the natural next step in the group testing literature. 

\subsection{Test design and notation} \label{sec:design}
To formalize our notation, we write $n$ for the number of individuals in the population, $\SIGMA$ for a binary vector representing the infection status of each individual, $k$ (the Hamming weight of $\SIGMA$) for the number of infected individuals and $m$ for the number of tests performed. 
We assume that $k$ is known for the purposes of matrix design, though in practice (see \cite[Remark 2.3]{Aldridge_2019_2}) it is generally enough to know $k$ up to a constant factor to design a matrix with good properties. In this paper, in line with other work such as \cite{Aldridge_2014}, we consider a scaling  $k\sim n^{\theta}$ for some fixed $\theta \in (0,1)$, referred to in \cite[Remark 1.1]{Aldridge_2019_2} as the sparse regime\footnote{\gebA{Note that the analysis directly extends to $k=\Theta(n^{\theta})$ as a constant factor in front does not influence the analysis.}}. 
In addition to the interesting phase transitions observed using this scaling, this sparse regime is particularly relevant as it was found suitable to model the early state of a pandemic \cite{Wang_2011}.

Let us next introduce the test design. With $V=(x_i)_{i \in [n]}$ denoting the set of $n$ individuals\footnote{$[n]$ will be used as an abbreviated notation for the set $\cbc{1, \dots, n}$.}   and $F=(a_i)_{i \in [m]}$ the set of $m$ tests, the test design can be envisioned as a bipartite factor graph with $n$ variable nodes "on the left" and $m$ factor nodes "on the right". We draw a configuration $\SIGMA \in \cbc{0,1}^V$,  encoding the infection status of each individual, uniformly at random from vectors of Hamming weight $k$. The set of healthy individuals will be denoted by $V_0$ and the set of infected individuals by $V_1$. In symbols,
\begin{align*}
	V_0 = \cbc{x \in V: \SIGMA(x) = 0} \qquad \text{and} \qquad V_1 = V \setminus V_0 = \cbc{x \in V: \SIGMA(x) = 1}
\end{align*}
The lower bound from \eqref{count} suggests that in the noisy group testing setting it is natural to compare the performance of algorithms and matrix designs in terms of  the prefactor of $k \log(n/k)$ in the number of tests required. 
To be precise, we carry out $m$ tests, and each item is assigned to exactly $\Delta$ tests chosen uniformly at random without replacement. We parameterize $m$ and $\Delta$ as
\begin{align} \label{eq:param}
	m = c k \log(n/k) \qquad \text{and} \qquad \Delta = cd\log(n/k)
\end{align}
for some suitably chosen constants $c,d \geq 0$. 
 
Let $\partial x$ denote the set of  tests that individual $x$ appears in and $\partial a$ the set of individuals assigned to test $a$. The resulting (non-constant) collection of test degrees will be denoted by the vector $\vec{\Gamma}=(\vec{\Gamma}_a)_{a \in [m]}$. Further, let
\begin{align} \label{eq:testdegrees}
	\Gamma_{\min} = \min_{a \in [m]} \Gamma_a \qquad \text{and} \qquad \Gamma_{\max} = \max_{a \in [m]} \Gamma_a. 
\end{align}
Throughout, $\G=\G(n,m,\Delta)$ describes the random bipartite factor graph from this construction.

Now consider the outcome of the tests. Recall from above that a standard noiseless group test $a$ gives a positive result if and only if there is at least one defective item contained in the pool, or equivalently if $\sum_{x \in \partial a} \SIGMA(x) \geq 1$. Even in the noisy case, this sum is a useful object to consider. Writing $\vecone$ for the indicator function, we define
\begin{equation} \label{eq:truesigma}
\SIGMA^*(a) = \vecone \cbc{ \sum_{x \in \partial a} \SIGMA(x) \geq 1 } \end{equation}
to be the outcome we would observe in the noiseless case using the test matrix corresponding to $\G$. We will say that test $a$ is {\em truly positive} if $\SIGMA^*(a) = 1$ and truly negative otherwise.

However, we do not observe the values of $\SIGMA^*(a)$ directly, but rather see what we will refer to as the {\em displayed} test outcomes $\hat \SIGMA(a)$ -- the outcomes of sending the true outcomes $\SIGMA^*(a)$ independently through the $p-q$ channel of Figure \ref{pqnoise}.  Since in this model a truly positive test remains positive with probability $1-q$ and a truly negative test is displayed as positive with probability $p$ we can write
\begin{align}
	\hat \SIGMA(a) & = \vecone\cbc{\Be(p)=1 } \left( 1 - \SIGMA^*(a) \right) + \vecone\cbc{ \Be(1-q)=1} \SIGMA^*(a) \label{eq:noisysigma}
\end{align}
where $\Be(r)$ denotes a Bernoulli random variable with parameter $r$ independent of all other randomness in the model. For models with binary outputs, this is the most general channel satisfying the noisy defective channel property of \cite[Definition 3.3]{Aldridge_2019_2}, though more general models are possible under the only defects matter property \cite[Definition 3.2]{Aldridge_2019_2}, where the probability of a test being positive depends on the number of  infected individuals it contains.

Note that if $p + q > 1$, we can preprocess the outputs from \eqref{eq:noisysigma} by flipping them, i.e. setting $\wt{p} = 1-p$ and $\wt{q} = 1-q$, where $\wt{p} + \wt{q} < 1$. Hence without loss of generality we will assume throughout that $p+q < 1$. In the case $p+q = 1$, the test outcomes are independent of the inputs, and we cannot hope to find the infected individuals -- see Corollary \ref{thm:shanCAP}.

With $\vm_0$ being the number of truly negative tests, let $\vm_0^f$ be the number of truly negative tests that are flipped to display a positive test result and $\vm_0^u$ be the number of truly negative tests that are unflipped. Similarly, define $\vm_1$ as the number of truly positive tests, of which $\vm_1^f$ are flipped to a negative test result and of which $\vm_1^u$ are unflipped. For reference, for $t \in \cbc{0,1}$ we write
\begin{align*}
\vm_t = &\abs{ \cbc{ a: \SIGMA^*(a) = t}} \\    
\vm_t^f = \abs{\cbc{ a: \SIGMA^*(a) = t, \hat \SIGMA(a) \neq t
}} &\quad \text{and} \quad    
\vm_t^u = \abs{ \cbc{ a: \SIGMA^*(a) = t, \hat \SIGMA(a) = t} } 
\end{align*}
\geb{Here we use bold letters to indicate random variables.}
Throughout the paper, we use the standard Landau notation $o(\cdot), O(\cdot), \Theta(\cdot), \Omega(\cdot), \omega(\cdot)$ and define $0 \log 0 = 0$.
Furthermore we say that a property $\cP$ holds \textit{with high probability ( \whp )}, if $\Pr \bc{ \cP} = 1$ as $n \to \infty$.
In order to quantify the performance of our algorithms, for any $0 < r \neq s < 1$, we write
\begin{align}\label{KL_definition}
   \KL{r}{s} :=  r \log \left( \frac{r}{s} \right) + (1-r) \log \left( \frac{1-r}{1-s} \right), 
\end{align}
for the relative entropy of a Bernoulli random variable with parameter $r$ to a Bernoulli random variable with parameter $s$, commonly referred to as the Kullback--Leibler divergence. Here and throughout the paper we use $\log$ to denote the natural logarithm. For $r$ or $s$ equal to $0$ or $1$ we define the value of $\KL{\cdot}{\cdot}$ (possibly infinite) on grounds of continuity, so for example $\KL{0}{s} = -\log(1-s)$.

\section{Main results}

With the test design and notation in place, we are now in a position to state our main results. \Thm s~\ref{thm_COMP}, \ref{thm_DD} 
are the centerpiece of this paper, featuring improved bounds for the noisy group testing problem for the general $p-q$ model. We follow up in Section \ref{sec:comb} with a discussion of the combinatorics underlying both algorithms, and provide a converse bound in Section \ref{sec:channel}. Subsequently, in Section \ref{sec:standard} we show how the bounds simplify when we consider the special cases of the Z, the reverse Z and Binary Symmetric Channel. Finally, in Section \ref{sec:better} we derive sufficient conditions under which {\tt DD}\ 
 requires fewer tests than the {\tt COMP} algorithm and compare the bounds of our constant-column design against the Bernoulli design employed in prior literature.

\subsection{Bounds for Noisy Group Testing}

We will consider two well-known algorithms from the noiseless setting to identify infected individuals in this paper. First, we study a noisy variant of the {\tt COMP} algorithm, originally introduced in \cite{Chan_2011}.

\VS
\begin{algorithm}[H] 
Declare every individual that appears in  $\alpha \Delta$ or more displayed negative tests as healthy. \\
Declare all remaining individuals as infected.
\caption{The noisy {\tt COMP} algorithm}
\label{comp_algorithm}
\end{algorithm}
\VS

Note that \change{for $\alpha\Delta=1$ the formulation of Algorithm~\ref{comp_algorithm} coincides with} 
the standard ${\tt COMP}$ algorithm where an individual is classified as healthy if it appears in at least one displayed negative test which constitutes a sufficient condition in the noiseless case.
We now state the first main result of this paper.

\begin{theorem}[Noisy {\tt COMP}] \label{thm_COMP}
Let $p,q\geq 0$, $p+q<1, d \in (0,\infty), \alpha \in (q, e^{-d}(1-p) + \bc{1-e^{-d}}q)$. Suppose that $0<\theta<1$  and let
\begin{align*}
    \mcomp &= \mcomp(n,\theta, p, q) = \min_{\alpha, d} \max \cbc{b_1(\alpha, d), b_2(\alpha, d)} k \log(n/k) \\
&\text{where} \qquad b_1(\alpha, d) = \frac{\theta}{1-\theta} \frac{1}{d \KL{\alpha}{q}} \\
&\text{and} \qquad b_2(\alpha, d) = \frac{1}{1-\theta}\frac{1}{d \KL{\alpha}{e^{-d}(1-p) + \bc{1-e^{-d}}q}}
\end{align*}
If $m \geq (1+\eps) \mcomp$ for some $\eps>0$, noisy {\tt COMP} will recover $\SIGMA$ \whp\ given test design $\G$ and test results $\hat \SIGMA$.
\end{theorem}

The noisy variant of the {\tt DD} algorithm of \cite{Aldridge_2014} was introduced in \cite{Johnson_2018} and reads as follows:

\VS
\begin{algorithm}[H]
Declare every individual that appears in  $\alpha \Delta$ or more displayed negative tests as healthy and remove such individual from every assigned test. \\
Declare every yet unclassified individual who is now the only unclassified individual in $\beta \Delta$ or more displayed positive tests as infected. \\
Declare all remaining individuals as healthy.
\caption{The noisy {\tt DD} algorithm \cite{Johnson_2018}}
\label{dd_algorithm}
\end{algorithm}
\VS

\change{Note that the formulation of Algorithm~\ref{dd_algorithm}} reduces to the noiseless version of {\tt DD} introduced in \cite{Aldridge_2014} by taking \change{$\alpha\Delta=\beta\Delta=1$. This is because in the noiseless setting a single negative test or a single positive test with just individuals already classified as uninfected is sufficient in the noiseless case}. 
\gebA{Furthermore note that for $\beta=0$ noisy {\tt DD} and noisy {\tt COMP} are the same. From now on we assume $\beta>0$.} \gebA{The proof of Theorem~\ref{thm_COMP} can be found in Appendix~\ref{sec_COMP}}. We now state the second main result of the paper. 

\begin{theorem}[Noisy {\tt DD}] \label{thm_DD}
Let $p,q\geq0$, $p+q<1, d \in (0, \infty), \alpha \in (q, e^{-d}(1-p) + \bc{1-e^{-d}}q)$ and $\beta \in (0, e^{-d}(1-q))$ and define $w = e^{-d}p + (1-e^{-d}) (1-q)$.
Suppose that $0<\theta<1$  and let
\begin{align*}
   \mdd &= \mdd(n,\theta, p, q) = \min_{\alpha, \beta, d} \max \cbc{c_1(\alpha, d), c_2(\alpha, d), c_3(\beta, d), c_4(\alpha, \beta, d)} k \log(n/k) \\
&\text{where} \qquad c_1(\alpha, d) = \frac{\theta}{1-\theta} \frac{1}{d \KL{\alpha}{q}} \\
&\text{and} \qquad c_2(\alpha, d) = \frac{1}{d \KL{\alpha}{1-w}} \\
&\text{and} \qquad c_3(\beta, d) = \frac{\theta}{1-\theta} \frac{1}{d \KL{\beta}{(1-q) e^{-d}}} \\
&\text{and} \qquad c_4(\alpha, \beta, d) = \max_{1-\alpha\leq z \leq 1} \cbc{ \frac{1}{1-\theta} \frac{1}{d \bc{ \KL{z}{w} + \vecone \cbc{\beta > \frac{z e^{-d}p}{w}} z \KL{\frac \beta z}{\frac{e^{-d}p}{w}}}}}
\end{align*}
If $m \geq (1+\eps) \mdd$ for some $\eps>0$, then noisy {\tt DD} will recover $\SIGMA$ \whp\ given test design $\G$ and test results $\hat \SIGMA$.
\end{theorem}

\gebA{The proof of Theorem~\ref{thm_DD} can be found in Appendix~\ref{sec_DD}}.
While the bounds appear cumbersome at first glance, the optimization is of finite dimension and for every specific value of $p$ and $q$ can be efficiently solved to arbitrary precision yielding explicit values for $\mcomp$ and $\mdd$. For illustration purposes, we will calculate those bounds for several values of $p,q$ and $\theta$.

\subsection{The combinatorics of the noisy group testing algorithms} \label{sec:comb}

In the following, we outline the combinatorial structures that Algorithm \ref{comp_algorithm} and \ref{dd_algorithm} take advantage of. \\
\geb{We start with defining the three types of tests that are relevant for the classification of an individual $x_i$  while using {\tt COMP} and {\tt DD}.}
\gebA{In the first stage we find
\begin{itemize}
    \item Type DN: Displayed negative tests
    \item Type DP: Displayed positive tests
\end{itemize}
}
Note that the only available information during the first stage of the algorithms is the test result and the pooling structure -- no information about the individuals' infection status is available. We give an illustration on the left hand side of Figure~\ref{neighbor}. After this step {\tt COMP} terminates by declaring all remaining individuals as infected. 

The {\tt DD} algorithm continues with a second step which considers just the displayed positive tests. From the first step of the algorithm one receives the estimate of the set of non-infected individuals obtained in the first round. Now distinguish the following two types, illustrated on the right hand side in Figure~\ref{neighbor}:
\begin{itemize}
    \item \gebA{Type \change{Displayed-Positive-Single} (DP-S)}: Displayed positive tests in which all other individuals are already declared as uninfected.
    \item \gebA{Type \change{Displayed-Positive-Multiple} (DP-M)}: Displayed positive tests with at least one other individual that is not contained in the estimated set of uninfected individuals.
\end{itemize}

\subsubsection{The noisy COMP algorithm}  \label{sec_COMP_explanation}
To get started, let us shed light on the combinatorics of noisy {\tt COMP} (Algorithm \ref{comp_algorithm}). For the \textit{noiseless} case, the {\tt COMP} algorithm classifies each individual that appears in at least one negative test as healthy and all other individuals as infected, since the participation in a negative test is a sufficient condition for the individual to be healthy. 

For the noisy case, the situation is not as straightforward, since an infected individual might appear in \textit{displayed} negative tests that were flipped when sent through the noisy channel. Thus, a single negative test is not definitive evidence that an individual is healthy. Yet, we can use the number of negative tests to tell the infected individuals apart from the healthy individuals. 

Clearly, noisy {\tt COMP} (Algorithm~\ref{comp_algorithm}) using a threshold $\alpha \Delta$ succeeds if no healthy individual appears in fewer than $\alpha \Delta$ displayed negative tests and no infected individual appears in more than $\alpha \Delta$ displayed negative tests. To this end, we define
\begin{align}
	\vN_x = \abs{\cbc{a \in \partial x: \hat \SIGMA(a) = 0}} \label{eq:Nxdef}
\end{align}
for the number of displayed negative tests that item $x$ appears in.
In terms of Figure \ref{neighbor}, the algorithm determines the infection status by counting the number of tests of \gebA{ Type DN}.

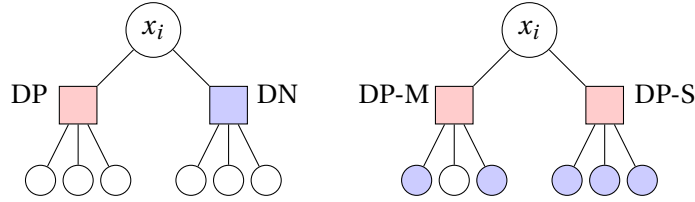
\begin{figure}[!ht]
\centering
\begin{tikzpicture}
\node[circle, draw, minimum width=0.5cm, fill=white] (xi) at (0,0) {$ x_i$};
\node[rectangle, draw, minimum width=0.5cm, minimum height=0.5cm, fill=blue!20] (a1) at (1, -1) {$ $};
\node(A) at (1.65,-0.85){\gebA{DN}};
\node[rectangle, draw, minimum width=0.5cm, minimum height=0.5cm, fill=red!20] (a2) at (-1, -1) {$ $};
\node(B) at (-1.65,-0.85){\gebA{DP}};
\node[circle, draw, minimum width=0.4cm, fill=white] (x1) at (0.5,-2) {$ $};
\node[circle, draw, minimum width=0.4cm, fill=white] (x2) at (1,-2) {$ $};
\node[circle, draw, minimum width=0.4cm, fill=white] (x3) at (1.5,-2) {$ $};
\node[circle, draw, minimum width=0.4cm, fill=white] (x4) at (-0.5,-2) {$ $};
\node[circle, draw, minimum width=0.4cm, fill=white] (x5) at (-1,-2) {$ $};
\node[circle, draw, minimum width=0.4cm, fill=white] (x6) at (-1.5,-2) {$ $};

\path[draw] (x1) -- (a1);
\path[draw] (x2) -- (a1);
\path[draw] (x3) -- (a1);
\path[draw] (x4) -- (a2);
\path[draw] (x5) -- (a2);
\path[draw] (x6) -- (a2);
\path[draw] (xi) -- (a1);
\path[draw] (xi) -- (a2);

\node[circle, draw, minimum width=0.5cm, fill=white] (xj) at (5,0) {$ x_i$};
\node[rectangle, draw, minimum width=0.5cm, minimum height=0.5cm, fill=red!20] (a3) at (6, -1) {$ $};
\node[rectangle, draw, minimum width=0.5cm, minimum height=0.5cm, fill=red!20] (a4) at (4, -1) {$ $};
\node[circle, draw, minimum width=0.4cm, fill=blue!20] (x7) at (5.5,-2) {$ $};
\node[circle, draw, minimum width=0.4cm, fill=blue!20] (x8) at (6,-2) {$ $};
\node[circle, draw, minimum width=0.4cm, fill=blue!20] (x9) at (6.5,-2) {$ $};
\node[circle, draw, minimum width=0.4cm, fill=blue!20] (x10) at (3.5,-2) {$ $};
\node[circle, draw, minimum width=0.4cm, fill=white] (x11) at (4,-2) {$ $};
\node[circle, draw, minimum width=0.4cm, fill=blue!20] (x12) at (4.5,-2) {$ $};
\node(C) at (6.8,-0.85){\gebA{DP-S}};
\node(D) at (3.2,-0.85){\gebA{DP-M}};
\path[draw] (x7) -- (a3);
\path[draw] (x8) -- (a3);
\path[draw] (x9) -- (a3);
\path[draw] (x10) -- (a4);
\path[draw] (x11) -- (a4);
\path[draw] (x12) -- (a4);
\path[draw] (xj) -- (a3);
\path[draw] (xj) -- (a4);
\end{tikzpicture}
\caption{\gebA{The relevant neighborhood structures for the analysis of the algorithms, on the left for the first stage and on the right for the second step. Rectangles represent tests (displayed positive in red, displayed negative in blue). Blue circles represent individuals that have been classified as healthy in the first step of {\tt DD} (or by {\tt COMP}). White circles represent individuals that are unclassified in the current stage. We refer to displayed negative tests as Type DN, displayed positive tests as Type DP, displayed positive with a single unclassified individual as Type DP-S and displayed positive with a multiple unclassified individual as Type DP-M}}
\label{neighbor}
\end{figure}

\subsubsection{The noisy {\tt DD} algorithm} \label{sec_DD_explanation}
As in the prior section, let us first consider the \textit{noiseless} {\tt DD} algorithm. The first step is identical to {\tt COMP} classifying all individuals that are contained in at least one negative test as healthy. In a second step, the algorithm checks each individual to see if \geb{ it is } contained in a positive test as the only \gebA{remaining} unclassified individual \gebA{after the first step of the algorithm} and thus must be infected. 

Again, the situation is more intricate when we add noise, since neither a single negative test gives us confidence that an individual is healthy nor does a positive test where the individual is the single \gebA{remaining} unclassified individual \gebA{after the first step of the algorithm} inform us that this individual must be infected. Instead we count and compare the number of such tests.
The first step of the noisy {\tt DD} algorithm is identical to noisy {\tt COMP}, but we are not required to identify all healthy individuals in the first step (\gebA{we are able to keep some unclassified for the second round}). Thus, after the first step, we are left with all infected individuals $V_1$ \geb{(as the algorithm did not \gebA{try to classify any individual as infected } in the first step)} and a set of yet unclassified healthy individuals \geb{(as some of them might exhibit a first neighbourhood that is not sufficient for a clear first round classification)} which we will denote by $\zeroplus$.These are healthy individuals who did not appear in sufficiently many displayed negative tests to be declared healthy with confidence in the first step\footnote{\gebA{Note that the bounds are taken in a way such that no infected individual is classified as uninfected in the first round.}}. In symbols, for some $\alpha \in (0,1)$
\begin{align*}
	\zeroplus = \cbc{x \in V_0: \vN_x < \alpha \Delta}
\end{align*}
To tell $V_1$ and $\zeroplus$ apart, we consider the number of displayed positive tests $\vP_x$ where the individual $x$ appears on its own after removing the 
 individuals \gebA{, which were declared healthy already, }$V_0 \setminus \zeroplus$ from the first step, i.e.
\begin{align}
	\vP_x = \abs{\cbc{a \in \partial x: \hat\SIGMA(a) = 1 \text{ and } \partial a \setminus \cbc{x} \subset V_0 \setminus \zeroplus}} \label{eq:Pxdef}
\end{align}
Referring to Figure \ref{neighbor}, the second step of the algorithm is based on counting tests of \gebA{Type DP-S}. Tests of \gebA{Type DP-M} contain another \gebA{remaining} unclassified individual \gebA{after the first step of the algorithm} from $V_{0,PD}\cup V_1$. 
The noisy {\tt DD} algorithm takes advantage of the fact that it is less likely for an individual $x \in \zeroplus$ to appear as the only yet unclassified individual in a displayed positive test than it is for an individual in $x \in V_1$. For $x \in \zeroplus$ such a test would be truly negative and would have been flipped (which occurs with probability $p$) to display a positive test result. Conversely, an individual $x \in V_1$ renders any of its tests truly positive and thus the only requirement is that the test otherwise contains only individuals \gebA{which were declared healthy already, }and is not flipped (which occurs with probability $1-q$). 
For this reason, we will see that the distribution of $\vP_x$ differs between $x \in V_1$ and $x \in \zeroplus$, and the difference $(1-q)-p > 0$ helps determine the size of this difference.  The second step of {\tt DD} exploits this observation by counting tests of \gebA{Type DP-S}.

\subsection{The Channel Perspective of noisy group testing} \label{sec:channel}
Motivated by \eqref{count}, we can describe the bounds in terms of rate, in a Shannon-theoretic sense. That is, we \gebA{follow the common notion to define} the rate (bits learned per test) of an algorithm in this setting \gebA{(for instance as in \cite{Baldassini_2013})} to be
\begin{align*}
    R := \frac{ \log \binom{n}{k}}{m \log 2} \sim \frac{k \log(n/k)}{m \log 2}.
\end{align*}
(Recall that we take logarithms to base $e$ throughout this paper). For example the fact that Theorems \ref{thm_COMP} and \ref{thm_DD} show that noisy {\tt COMP} and {\tt DD} respectively can succeed \whp\ ; with $m \geq (1+\epsilon) c k \log(n/k)$ tests for some $c$ is equivalent to the fact that $R = 1/(c \log 2)$ is an achievable rate in a Shannon-theoretic sense.

We now give a counterpart to these two theorems by stating a universal converse for the $p-q$ channel below, improving on the universal counting bound from \eqref{count}. 
The starting observation (see \cite[Theorem 3.1]{Aldridge_2019_2}) is that no group testing algorithm can succeed \whp\ with rate greater than $\Cchan$,  the Shannon capacity of the corresponding noisy communication channel. Thus, we cannot hope to succeed \whp\ with $m < (1-\epsilon) c k \log(n/k)$ tests where $c = 1/( \Cchan \log 2)$. Hence  as a direct consequence of the value of the channel capacity of the $p-q$ channel, we deduce the following statement.
\begin{corollary} \label{thm:shanCAP}
Let $p,q \geq 0$, $p + q < 1$ and $\epsilon > 0$, write $h(\cdot)$ for the binary entropy in nats (logarithms taken to base $e$) and $\phi = \phi(p,q) = (h(p) -h(q))/(1-p-q)$. If we define
$$ \mcount = \left( \frac{1}{ \KL{q}{1/(1+e^\phi)}} \right) k \log(n/k),$$
then for $m \leq (1-\epsilon) \mcount$ no algorithm can recover
$\SIGMA$ \whp\ for any matrix design.
\end{corollary}

\begin{remark} This result follows from Lemma \ref{lem:shanCAP} derived in Appendix \ref{Notes_on_Capacity} below. As discussed there, this derivation \geb{(combined with the fact that each test is negative with probability $e^{-d}$}) suggests a choice of density for the matrix:
$$ d = \dch = \log(1-p-q) - \log \left( \frac{1}{1+e^\phi} - q \right).$$
While a choice of $\Delta=c \cdot \dch \cdot \log(n/k)$ is not \gebA{necessarily} optimal, it may be regarded as a sensible heuristic that provides good rates for a range of $p$ and $q$ values. 
\end{remark}

\subsection{Applying the results to standard channels} \label{sec:standard}

With \Thm~\ref{thm_COMP} and \Thm~\ref{thm_DD} we derived achievable rates for the generalized \textit{p-q-model} (see Figure~\ref{pqnoise}). Prior research considered the Z channel where $p=0$ and $q>0$, the Reverse Z channel where $p>0$ and $q=0$ and the Binary Symmetric Channel with $p=q>0$.
These channels are common models in coding theory \cite{MCT_2007}, but are also often considered in medical applications \cite{Lalkhen_2008,Long_2018} concerned with taking \geb{imperfect} sensitivity ($q>0$), specificity ($p>0$) or both ($p>0$ and $q>0$) into account. \geb{As a consequence we also compare our results with the most recent results of Johnson and Scarlett \cite{Johnson_2018}.}
In the following section we will demonstrate how performance guarantees 
on these channels can directly be obtained from our main theorems. 

\subsubsection{Recovery of the noiseless model}
\change{Note that the bounds Corollary~\ref{cor_noiseless_COMP} and Corollary~\ref{cor_noiseless_DD} are already known \cite{Chan_2011, Johnson_2019}. We would like to give the reader an idea of how one can see that our cumbersome looking bounds relate to the more accessible bounds given for the noiseless case.}
First, we show the noiseless bounds can be simply recovered by \change{letting $p,q\rightarrow 0$}. In the noiseless setting, \change{it is sufficient, by definition of the algorithm, to set both $\alpha\Delta=1$ and $\beta\Delta=1$.} 
To see why, observe that in the absence of noise a single negative test is sufficient evidence that an individual is healthy. Conversely, a single positive test where the individual only appears with 
individuals \gebA{, which were declared healthy already, } implies that particular individual must surely be infected. As shown in \cite{Coja_2019} the optimal parameter choice for  the density parameter  $d$ in the constant-column design in the noiseless setting is $\log(2)$. 
Applying these values to \Thm~\ref{thm_COMP} we recover the noiseless bound for {\tt COMP}.\geb{These bounds were first stated in \cite{Chan_2011}.}
\change{\begin{corollary}[{\tt COMP} in the noiseless setting] \label{cor_noiseless_COMP}
Let $p,q\rightarrow 0$, $0<\theta<1$ and $\eps>0$. Further, let
\begin{align*}
    m_{{\tt COMP}, \text{noiseless}} = \frac{1}{(1-\theta) \log^2 2} k \log(n/k).
\end{align*}
Furthermore let $m_{\tt COMP}(n,\theta,p,q)$ be defined as in \Thm~\ref{thm_COMP}
Then we find 
$$m_{\tt COMP}(n,\theta,p,q)\underset{p,q\rightarrow 0}{\longrightarrow} m_{{\tt COMP}, \text{noiseless}}$$
\end{corollary}}
\change{\begin{proof}
We start by taking the bounds $b_1(\alpha, d) $ and $b_2(\alpha, d)$. To see how this boils down to $m_{{\tt COMP}, \text{noiseless}}$, we start with using the well-known fact that within the near constant column design $d=\log(2)$ is the optimal choice \cite{Coja_2019}. Now by taking both $p,q\rightarrow 0$ one realizes that $b_1(\alpha,\log(2))$ vanishes as $\log(p)\rightarrow -\infty$ as $p\rightarrow 0$. Turning our focus to the second bound we see that it boils down to
$$b_2(\alpha,\log(2)))=\frac{1}{(1-\theta)\log(2)}\frac{1}{\log(2)+\alpha\log(\alpha)+(1-\alpha)\log(1-\alpha)}$$
On the one hand we realize that $\alpha\log(\alpha)+(1-\alpha)\log(1-\alpha)$ is negative for all $\alpha\in(0,1)$. This leads to
$$b_2(\alpha,\log(2))>b_2(0,\log(2))$$
On the other hand we realize that in the noiseless case a single negative test is sufficient for a classification as uninfected. 
Therefore we may choose $\alpha>0$ sufficiently small. One indeed realizes that for each $\alpha$ we can choose $\eps:=\eps(\alpha)> 0$ appropriately, such that the bounds given in  \Thm~\ref{thm_COMP} recover the noiseless case.
\end{proof}}
We also recover the noiseless bounds for the {\tt DD} algorithm as stated in \cite{Johnson_2019}.
\change{
\begin{corollary}[{\tt DD} in the noiseless setting] \label{cor_noiseless_DD}
Let $p,q\rightarrow 0, 0<\theta<1$ and $\eps>0$. Further, let
\begin{align*}
   m_{{\tt DD}, \text{noiseless}} = \max \cbc{1, \frac{\theta}{1-\theta}} \frac{1}{\log^2 2} k \log(n/k).
\end{align*}
Furthermore let $m_{\tt DD}(n,\theta,p,q)$ be defined as in \Thm~\ref{thm_DD}
Then we find 
$$m_{\tt DD}(n,\theta,p,q)\underset{p,q\rightarrow 0}{\longrightarrow} m_{{\tt DD}, \text{noiseless}}$$
\end{corollary}
\begin{proof}
We start with taking $c_1(\alpha,d),c_2(\alpha,d),c_3(\beta,d)$ and $c_4(\alpha,\beta,d)$ as defined in Theorem~\ref{thm_DD}. First of all we take $c_4(\alpha,\beta,d)$. By assumption we find $\beta>0$ and therefore the indicator is 1 as soon as we let $p\rightarrow 0$. Furthermore for $p\rightarrow 0$ we get $-\log(p)\rightarrow \infty$ and find $c_4\rightarrow 0$. Second of all we take $c_1(\alpha,d)$. With a similar argument as before we see that $c_1(\alpha,d)\rightarrow 0$ for $q\rightarrow 0$ as in this case we find $-\log(q)\rightarrow \infty$. Therefore we are left with $c_2(\beta,d)$ and $c_3(\alpha,\beta,d)$. Again, we use the well known fact that in the noiseless case $d=\log(2)$ is the optimal choice. Therefore with $p,q\rightarrow 0$ the two remaining bounds read as follows:
\begin{align*}
    c_2(\alpha,\log(2))=\frac{1}{\log(2)\bc{\log(2)+\alpha\log(\alpha)+(1-\alpha)\log(1-\alpha)}}\\
    c_3(\alpha, \beta,\log(2))=\frac{\theta}{(1-\theta)}\frac{1}{\log(2)\bc{\log(2)+\beta\log(\beta)+(1-\beta)\log(1-\beta)}}
\end{align*}
Again we see that $x\log(x)+(1-x)\log(1-x)$ is negative for $x\in(0,1)$. Therefore we find
\begin{align*}
    c_2(\alpha,\log(2))>c_2(0,\log(2))\\
    c_3(\alpha,\log(2))>c_3(0,\log(2))
\end{align*}
Now as as before in this case again a single negative test as well as a single test with only already classified uninfected individuals is sufficient. Therefore we can choose $\alpha,\beta>0$ sufficiently small. One indeed realizes that for each $\alpha,\beta>0$ one can choose $\eps:=\eps(\alpha,\beta)$ appropriately such that the bounds of \Thm~\ref{thm_DD} recover the noiseless case.
\end{proof}
}

\subsubsection{The Z channel}

In the Z channel, we have $p=0$ and $q>0$, i.e. no truly negative test displays a positive test result. Thus, \gebA{ in this case finding one positive test with only one unclassified individual is a clear indication},  \change{therefore we again can choose $\beta>0$ sufficiently small} 
and remain agnostic about $\alpha$ and $d$. The bounds for {\tt COMP} and {\tt DD} thus read as follows.

\begin{corollary}[Noisy {\tt COMP} for the Z channel] \label{cor_Z_COMP}
Let $p\change{\rightarrow}0, 0<q<1, 0<\theta<1$ and $\eps>0$. Further, let
\begin{align*}
m_{{\tt COMP}, Z} &= \min_{\alpha, d} \max \cbc{b_1(\alpha, d), b_2(\alpha, d)} k \log(n/k)\\
   \text{ with } \quad b_1(\alpha, d) &= \frac{\theta}{1-\theta} \frac{1}{d \KL{\alpha}{q}} \quad \text{ and } \quad b_2(\alpha, d) = \frac{1}{1-\theta} \frac{1}{d \KL{\alpha}{e^{-d} + \bc{1-e^{-d}}q}}.
\end{align*}
If $m > (1+\eps) m_{{\tt COMP}, Z}$, noisy {\tt COMP} will recover $\SIGMA$ \whp\ given $\G, \hat \SIGMA$.
\end{corollary}

\begin{corollary}[Noisy {\tt DD} for the Z channel] \label{cor_Z_DD}
Let $p\change{\rightarrow}0, 0<q<1, 0<\theta<1$ and $\eps>0$. Further, let
\begin{align*}
m_{{\tt DD}, Z} &= \min_{\alpha, d} \max \cbc{c_1(\alpha, d), c_2(\alpha, d), c_3(d)} k \log(n/k)\\
    \text{ with } \quad &c_1(\alpha, d) = \frac{\theta}{1-\theta} \frac{1}{d \KL{\alpha}{q}}
    \quad \text{ and } \quad c_2(\alpha, d) = \frac{1}{d \KL{\alpha}{e^{-d}+\bc{1-e^{-d}}q}} \\
    \quad \text{ and } \quad &c_3(d) = \frac{\theta}{1-\theta} \frac{1}{-d\log \bc{1-e^{-d}(1-q)}}. 
\end{align*}
If $m > (1+\eps) m_{{\tt DD}, Z}$, noisy {\tt DD} will recover $\SIGMA$ \whp\ given $\G, \hat \SIGMA$.
\end{corollary}

\begin{proof}
The bounds $c_1$ and $c_2$ follow directly from \Thm~\ref{thm_DD} by  \change{letting $p\rightarrow 0$. An immediate consequence of $p\rightarrow 0$ is that due to the fact that $-\log(p)\rightarrow \infty$ and one finds that $c_4\rightarrow 0$, thus being trivial in this case.} For $c_3$ we use the fact that \change{we can choose $\beta>0$ sufficiently small  we find $\KL{\alpha}{e^{-d}(1-q)} = -\log \bc{1-e^{-d}(1-q)} - \delta(\beta)$ for $\delta(\beta)>0$.  Note that by definition of the noise model, we may choose an arbitrary $\beta_{\min}$ very close to zero and as a consequence $\beta=\beta_{\min}$ leading to $\delta(\beta)\rightarrow \delta_{\min}$. The assertion follows as for each $\beta$ we may choose $\eps:=\eps(\beta)>0$ such that $(1+\eps)>\bc{1+\eps\bc{\beta_{\min}}}$.}
\end{proof}

An illustration of the bounds from Corollary \ref{cor_Z_COMP} and \ref{cor_Z_DD} for sample values of $q$ is shown in Figure \ref{fig_Z_channel}.

\subsubsection{Reverse Z channel}

In the reverse Z channel, we have $q=0$ and $p>0$, i.e. no truly positive test displays a negative test result. Thus, we \change{may choose $\alpha>0$ sufficiently small }
and remain agnostic about $\beta$ and $d$. The bounds for the noisy {\tt COMP} and {\tt DD} thus read as follows.

\begin{corollary}[Noisy {\tt COMP} for the Reverse Z channel] \label{cor_RZ_COMP}
Let $0<p<1, q\rightarrow 0, 0<\theta<1$ and $\eps>0$. Further, let
\begin{align*}
    m_{{\tt COMP}, \text{rev Z}} &= \frac{1}{1-\theta} \min_d \cbc{\frac{1}{-d \log \bc{1- e^{-d}(1-p)}}} k \log(n/k).
\end{align*}
If $m > (1+\eps) m_{{\tt COMP}, \text{rev Z}}$, noisy {\tt COMP} will recover $\SIGMA$ \whp\ given $\G, \hat \SIGMA$.
\end{corollary}

\begin{proof}
The corollary follows from \Thm~\ref{thm_COMP} and the fact that \change{ for $q\rightarrow 0$ one finds that $\KL{\alpha}{0}$ diverges, Thereby $b_1\rightarrow 0$ just gives a trivial bound in this case. Furthermore for sufficiently small $\alpha>0$ we get $\KL{\alpha}{e^{-d}(1-p)} \rightarrow - \log \bc{1-e^{-d}(1-p)}-\delta(\alpha)$. Due to the noise assumption, we may choose an arbitrary $\alpha_{\min}$ very close to zero and $\alpha=\alpha_{\min}$ which leads to $\delta(\alpha)\rightarrow \delta\bc{\alpha_{\min}}$. The assertion follows by choosing $\eps:=\eps(\alpha)>0$ such that $(1+\eps)>\bc{1+\eps\bc{\alpha_{\min}}}$.}
\end{proof}
%
Note that Corollary \ref{cor_RZ_COMP} does not yield an immediate closed form expression for the optimal value of $d$.

\begin{corollary}[Noisy {\tt DD} in the Reverse Z channel]
\label{cor_RZ_DD}
Let \hspace{1mm} $0<p<1, q \rightarrow 0, 0<\theta<1$ and $\eps>0$. Further, let
\begin{align*}
 m_{{\tt DD}, \text{rev Z}} &= \min_{\beta, d} \max \cbc{c_2(d), c_3(\beta, d), c_4(\beta, d)} k \log(n/k)\\
    \text{ with } &c_2(d) = \frac{1}{-d \log \bc{1-e^{-d}(1-p)}}
    \quad \text{ and } \quad c_3(\beta, d) = \frac{\theta}{1-\theta} \frac{1}{d \KL{\beta}{e^{-d}}} \\
    \text{ and } &c_4(\beta, d) = \frac{1}{1-\theta} \frac{1}{d \bc{-\log\bc{1-e^{-d}(1-p)}+\KL{\beta}{\frac{e^{-d}p}{e^{-d}p+\bc{1-e^{-d}}}}}} 
\end{align*}
If $m > (1+\eps) m_{{\tt DD}, \text{rev Z}}$, noisy {\tt DD} will recover $\SIGMA$ \whp\ given $\G, \hat \SIGMA$.
\end{corollary}

\begin{proof}
\change{First of all we assume $q\rightarrow 0$. Therefore we find $c_1\rightarrow 0$ as $-\log(q)\rightarrow \infty$.} The bounds $c_2, c_3$ follow from \Thm~\ref{thm_DD} and the same manipulations as above. \change{ For $c_4$, we again see that by definition of the noise model we may choose $\alpha>0$ as close to zero as we like. Therefore we get $(1-\alpha)$ close to 1, which leads to $z\rightarrow 1$. The assertion follows as for each $\alpha$ we can choose $\eps:=\eps(\alpha)>0$ such that $(1+\eps)>\bc{1+\eps\bc{\alpha_{\min}}}$.}
\end{proof}

An illustration of the bounds of Corollary \ref{cor_RZ_COMP} and \ref{cor_RZ_DD} for sample values of $p$ is shown in Figure \ref{fig_rev_Z_channel}.

\subsubsection{Binary Symmetric Channel}

In the Binary Symmetric Channel (BSC), we set $p=q>0$. Even though information-theoretic arguments would suggest setting $d=\log 2$, we formulate the expression below with general $d$. We also keep the threshold parameters $\alpha$ and $\beta$. The bounds for the noisy {\tt DD} and {\tt COMP} only simplify slightly.

\begin{corollary}[Noisy {\tt COMP} in the Binary Symmetric Channel] \label{cor_BSC_COMP}
Let $0<p=q<1/2, 0<\theta<1$ and $\eps>0$. Further, let
\begin{align*}
 &m_{{\tt COMP}, \text{BSC}} = \min_{\alpha, d} \max \cbc{b_1(\alpha, d), b_2(\alpha, d)} k \log(n/k)\\
  \text{ with } \quad  &b_1(\alpha, d) = \frac{\theta}{1-\theta} \frac{1}{d \KL{\alpha}{p}}
  \quad \text{ and } \quad  b_2(\alpha, d) = \frac{1}{1-\theta} \frac{1}{d \KL{\alpha}{e^{-d}+p-2e^{-d} p}} 
   .
\end{align*}
If $m > (1+\eps) m_{{\tt COMP}, \text{BSC}}$, noisy {\tt COMP} will recover $\SIGMA$ \whp\ given $\G, \hat \SIGMA$.
\end{corollary}

\begin{corollary}[Noisy {\tt DD} in the Binary Symmetric Channel]
\label{cor_BSC_DD}
Let $0<p=q<1/2, 0<\theta<1$ and $\eps>0$ and define $v=1-e^{-d}-p+2e^{-d}p$. Further, let
\begin{align*}
 &m_{{\tt DD}, \text{BSC}} = \min_{\alpha, \beta, d} \max \cbc{c_1(\alpha, d), c_2(\alpha, d), c_3(\beta, d), c_4(\alpha, \beta, d)} k \log(n/k)\\
    \text{ with }\quad &c_1(\alpha, d) = \frac{\theta}{1-\theta} \frac{1}{d \KL{\alpha}{p}}
    \quad \text{ and } \quad c_2(\alpha, d) = \frac{1}{d \KL{\alpha}{e^{-d}+p-2e^{-d}p}} \\
    \text{ and }\quad &c_3(\beta, d) = \frac{\theta}{1-\theta} \frac{1}{d \KL{\beta}{(1-p) e^{-d}}} \\
    \text{ and } \quad &c_4(\alpha, \beta, d) = \max_{1-\alpha \leq z \leq 1} \cbc{ \frac{1}{1-\theta} \frac{1}{d \bc{ \KL{z}{v} + \vecone \cbc{\beta > \frac{ze^{-d}p}{v}} z \KL{\frac \beta z}{\frac{e^{-d}p}{v}}}}} 
   .
\end{align*}
If $m > (1+\eps) m_{{\tt DD}, \text{BSC}}$, noisy {\tt DD} will recover $\SIGMA$ \whp\ given $\G, \hat \SIGMA$.
\end{corollary}

An illustration of the bounds of Corollary \ref{cor_BSC_COMP} and \ref{cor_BSC_DD} is shown in Figure \ref{fig_BS_channel}.

\subsection{Comparison of  noisy {\tt COMP} and {\tt DD}} \label{sec:better}

An obvious next question is to find conditions under which the noisy {\tt DD} algorithm requires fewer tests than the noisy {\tt COMP}. For the noiseless setting, it can be easily shown that {\tt DD} provably outperforms {\tt COMP} for all $\theta \in (0,1)$. For the noisy case, matters are slightly more complicated. 

Recall that noisy {\tt COMP} classifies all individuals appearing in less than $\alpha \Delta$ displayed negative tests as infected while noisy {\tt DD} additionally requires such individuals to appear in more than $\beta \Delta$ displayed positive tests as the only yet unclassified individual. Thus, it might well be that an infected individual is classified correctly by noisy {\tt COMP}, while it is missed by the noisy {\tt DD} algorithm. 

That being said, our simulations indicate that noisy {\tt DD} generally \gebA{requires fewer tests than} noisy {\tt COMP}, but for the reason mentioned above we can only prove that  for the reverse Z channel while remaining agnostic about the Z channel and the Binary Symmetric Channel, as the next proposition evinces.
\begin{proposition} \label{prop_comp_dd}
For all $p,q \geq 0$ with $p+q<1$ there exists a $d^* \in (0,\infty)$ such that $\mcomp \geq \mdd$ as long as $e^{-d^*}p \geq q$.
\end{proposition}
In terms of the common noise channels \Prop~\ref{prop_comp_dd} gives the following corollary.
\begin{corollary}
In the reverse Z channel, $\mcomp \geq \mdd$.
\end{corollary}
\geb{The proof can be found in Appendix~\ref{sec_Com_DD_COMP}.} Our simulations suggest that this superior performance of noisy {\tt DD} holds as well for the Z channel and Binary Symmetric Channel. Please refer to Figure~\ref{COMPvsDD} for an illustration.

\begin{figure}
    \centering
    \begin{subfigure}{.4\textwidth}
    \centering
    \includegraphics[width=1\linewidth]{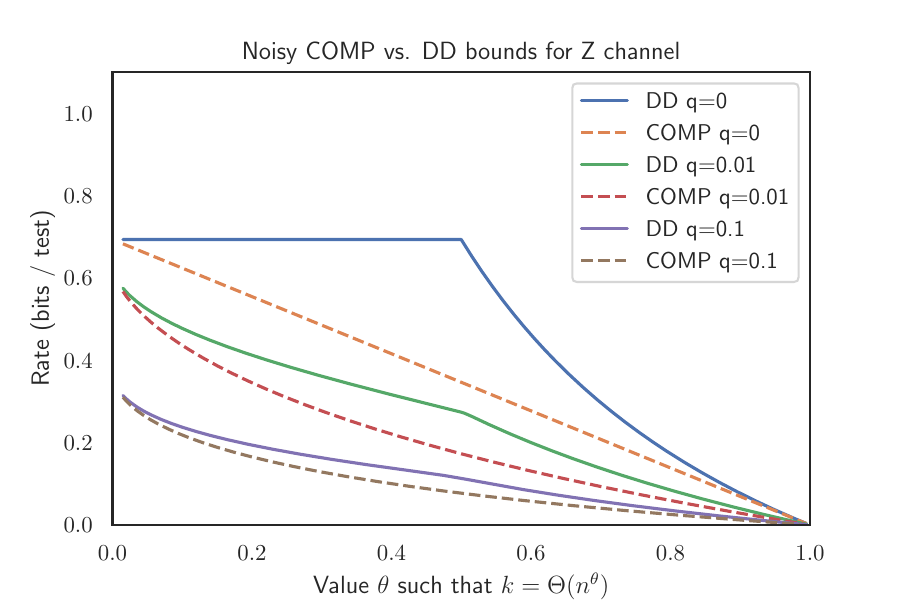}
    \end{subfigure}
    \begin{subfigure}{.4\textwidth}
    \centering
    \includegraphics[width=1\linewidth]{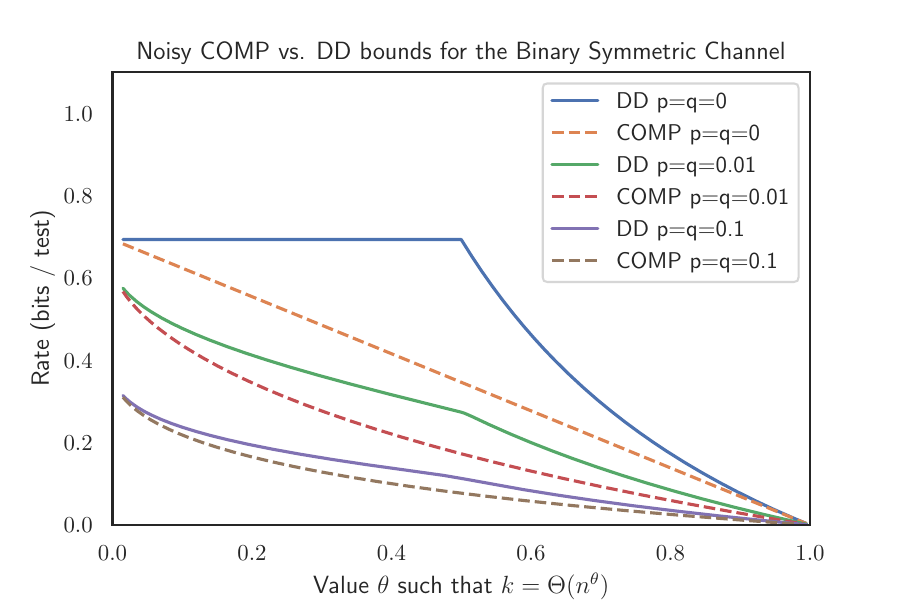}
     \end{subfigure}
     \caption{Comparison of the bound for noisy {\tt DD} and noisy {\tt COMP} in the Z-channel and the Binary Symmetric Channel for different noise level. \geb{(Note for black and white prints: The lines in the diagram are in the same order as given in the legend from top to bottom) }}
     \label{COMPvsDD}
\end{figure}

\subsection{Relation to Bernoulli testing}
\geb{In \cite{Johnson_2018} sufficient bounds for noisy group testing and a Bernoulli test design where each individual joins every test independently with some fixed probability were derived.}
Thus, the variable degrees fluctuate and we end up with some individuals assigned only to few tests.
In contrast, we work under a model in this paper where each individual joins an equal number of tests $\Delta$ chosen uniformly at random without replacement.
For the noiseless case, it is by now clear that the \gebA{near}-constant-column design better facilitates inference than the Bernoulli test design \cite{Coja_2019, Johnson_2019}.
We find that the same holds true for the noisy variant of the {\tt COMP} algorithm. Let us denote by $\mcomp^{\text{Ber}}$ the number of tests required for the noisy {\tt COMP} to succeed under a Bernoulli test design.

\begin{proposition} \label{prop_Ber_COMP}
For all $p+q < 1$, we have
\begin{align*}
    \mcomp^{\text{Ber}} \geq \mcomp 
\end{align*}
\end{proposition}

We see the same effect for the noisy variant of the {\tt DD} algorithm for all simulations, but for technical reasons only  prove it for the Z channel.

\begin{proposition} \label{prop_Ber_DD}
For the Z channel where $p=0$ and $0<q<1$, we have
\begin{align*}
    \mdd^{\text{Ber}} > \mdd 
\end{align*}
\end{proposition}

For an illustration on the magnitude of the difference, we refer to Figure~\ref{FigBervsCC} and Figure~\ref{fig_Z_BervsCC}.

\begin{figure}[!ht]
    \centering
    \begin{subfigure}{.4\textwidth}
    \centering
    \includegraphics[width=1\linewidth]{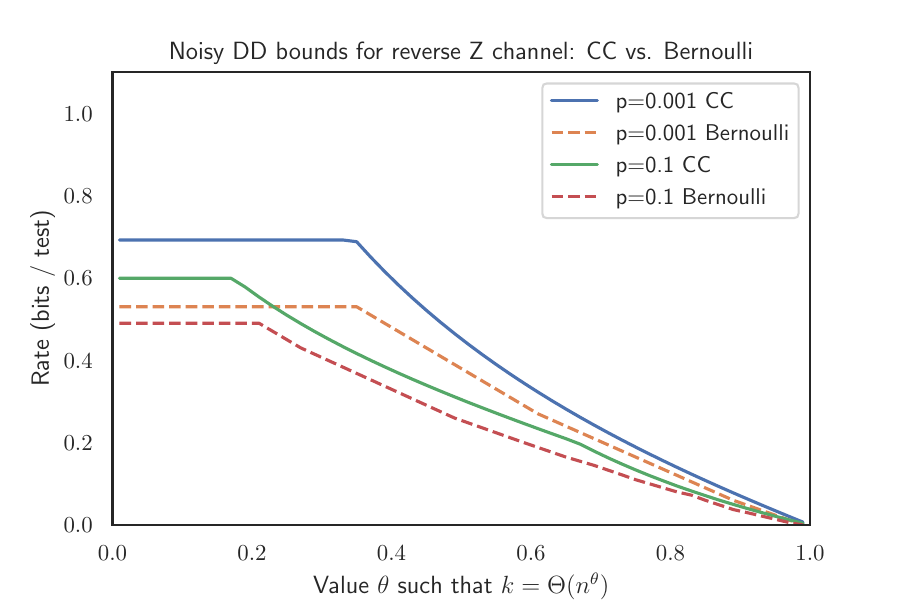}
    \end{subfigure}
    \begin{subfigure}{.4\textwidth}
    \centering
    \includegraphics[width=1\linewidth]{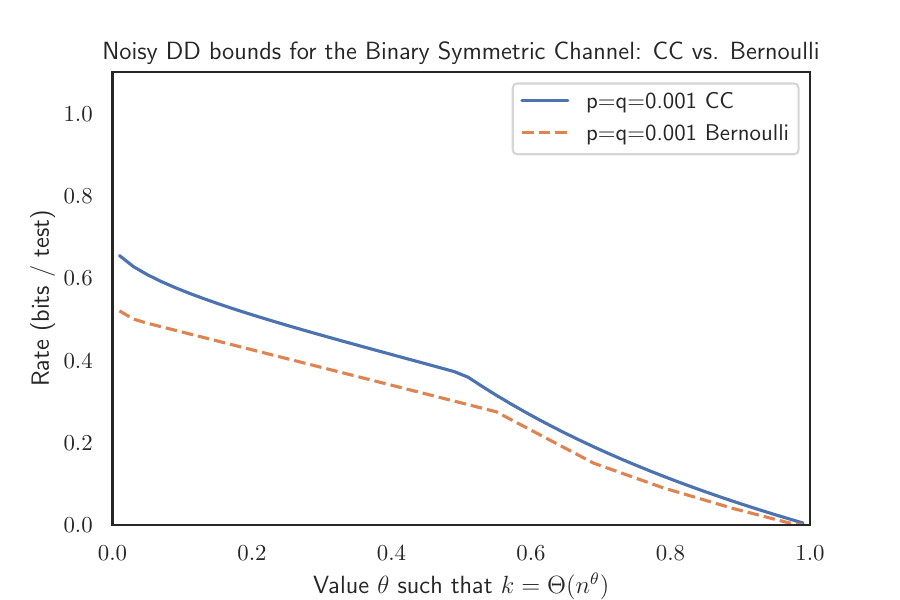}
     \end{subfigure}
     \caption{Comparison of {\tt DD} bounds under a Bernoulli test design (\cite{Johnson_2018}) and constant column test design (present paper) for the reverse Z and Binary Symmetric Channel. \geb{(Note for black and white prints: The solid lines as well as the dashed lines in the diagram are in the same order as given in the legend from top to bottom) }}
     \label{FigBervsCC}
     
\end{figure}

\newpage

\appendix

\section*{Appendix}

The core of the technical sections is the proof of \Thm s~\ref{thm_COMP} and \Thm~\ref{thm_DD}. Some groundwork with standard concentration bounds and group testing properties can be found in \Sec~\ref{sec_groundwork}. We continue with the proof of \Thm s~\ref{thm_COMP} and \ref{thm_DD} in \Sec s~\ref{sec_COMP} and \ref{sec_DD}, respectively. The structure of the proofs follows a similar logic. First, we derive the distributions for the number of displayed positive and negative tests for infected and healthy individuals. Second, we threshold these distributions using sharp Chernoff concentration bounds to deduce the bounds stated in \Thm~\ref{thm_COMP} and \Thm~\ref{thm_DD}. 
Thereafter, we proceed to the proof of \Prop~\ref{prop_comp_dd} in \Sec~\ref{sec_Com_DD_COMP}, while the proofs of \Prop s~\ref{prop_Ber_COMP} and \ref{prop_Ber_DD} follow in \Sec~\ref{sec_Rel_Ber}. \gebA{The proof of Corollary~\ref{thm:shanCAP} can be found in \Sec~\ref{Notes_on_Capacity}. Additional illustrations of our results for the different channels can be found in \Sec~\ref{Appendix_Illustration}.}

\section{Groundwork} \label{sec_groundwork}

For starters, let us recall the Chernoff bound for binomial and hypergeometric distributions.

\begin{lemma}[Chernoff bound for the binomial distribution \cite{Janson_2011}] \label{lem_chernoff_bin}
Let $p<q<r \in (0,1)$ and $\vX \sim \Bin(n,q)$ be a binomially distributed random variable. Then
\begin{align*}
    \Pr \bc{\vX \leq \lceil pn \rceil} &= \exp \bc{-\bc{1+n^{-\Omega(1)}}n \KL{p}{q}} \\
    \Pr \bc{\vX \geq \lceil rn \rceil} &= \exp \bc{-\bc{1+n^{-\Omega(1)}}n \KL{r}{q}}
\end{align*}
\end{lemma}

\begin{lemma}[Chernoff bound for the hypergeometric distribution  \cite{Hoeffding_1963}] \label{lem_chernoff_hyp}
Let $p<q<r \in (0,1)$ and $\vY \sim H(N, Q, n)$ be a hypergeometrically distributed random variable. Further, let $q = Q/N$. Then
\begin{align*}
    \Pr \bc{\vY \leq \lceil pn \rceil} &= \exp \bc{-\bc{1+n^{-\Omega(1)}} n \KL{p}{q}} \\
    \Pr \bc{\vY \geq \lceil rn \rceil} &= \exp \bc{-\bc{1+n^{-\Omega(1)}} n \KL{r}{q}}
\end{align*}
\end{lemma}

The next lemma provides that the test degrees, as defined in \eqref{eq:testdegrees} above, are tightly concentrated. Recall from \eqref{eq:param} that the number of tests $m = c k \log(n/k)$ and each item appears in $\Delta = cd\log(n/k)$ tests.

\begin{lemma} \label{lem_gamma_minmax}
With probability $1-o(n^{-2})$ we have
\begin{align*}
    dn/k - \sqrt{dn/k} \log n \leq \vec{\Gamma}_{\min} \leq \vec{\Gamma}_{\max} \leq dn/k  \gebA{+} \sqrt{dn/k} \log n
\end{align*}
\end{lemma}
\begin{proof}
The probability that an individual $x$ is assigned to test $a$ is given by
\begin{align} \label{eq_prop_test}
    \Pr \bc{x \in \partial a} = 1 - \Pr \bc{x \notin \partial a} = 1 - \binom{m-1}{\Delta} \binom{m}{\Delta}^{-1} = \Delta/m = d/k
\end{align}
Since each individual is assigned to tests independently, the total number of individuals in a given test follows the binomial distribution $\Bin \bc{n, d/k }$. \gebA{The assertion now follows from applying the Chernoff bound for this binomial distribution at the expectation }  (\Lem~\ref{lem_chernoff_bin}).
\end{proof}

Next, we show that the number of truly negative tests $\vm_0$ (and thus the number of truly positive tests $\vm_1$) is tightly concentrated.

\begin{lemma} \label{lem_m0}
With probability $1-o(n^{-2})$ we have $\mzero = e^{-d} m + O(\sqrt{m} \log^3 n)$.
\end{lemma}

\begin{proof}
Recall from \eqref{eq_prop_test} that
\begin{align*}
    \Pr \bc{x \in \partial a} = d/k
\end{align*}
Since infected individuals are assigned to tests mutually independently, we find for a test $a$ that
\begin{align*}
    \Pr \bc{V_1 \cap \partial a = \emptyset} = \Pr \bc{\Bin \bc{k, d/k} = 0} = \bc{1- d/k }^k =\bc{1+n^{-\Omega(1)}} e^{-d}.
\end{align*}
Consequently, $\Erw \brk{\mzero} =\bc{1+n^{-\Omega(1)}} e^{-d} m$. Finally, changing the set of tests for a specific infected individual shifts the total number of negative tests by at most $\Delta$. Therefore, the \gebA{McDiarmid} inequality \gebA{(Lemma 1.2 in \cite{Mcdiarmid_1989})} yields

\begin{align*}
    \Pr \bc{\abs{\mzero - \Erw \brk{\mzero}} \geq t} \leq 2 \exp \bc{-\frac{t^2}{4k\Delta^2}}.
\end{align*}
The lemma follows from setting $t= O \bc{ \sqrt{m} \log^3 n}$.
\end{proof}

With the concentration of $\mzero$ and $\mone$ at hand, we readily obtain estimates for $\vm_0^f, \vm_0^u, \vm_1^f$ and $\vm_1^u$. \geb{We remind ourselves that these are the number of flipped, unflipped negative tests and the number of flipped, unflipped positive tests as defined in Sec.~\ref{sec:design}.}

\begin{corollary} \label{cor_m0_m1}
 With probability $1-o(n^{-2})$ we have 
\begin{itemize}
    \item[(i)] $\mzero^f = e^{-d}pm + O \bc{\sqrt{m} \log^3 n}$
    \item[(ii)] $\mzero^u = e^{-d}(1-p)m + O \bc{\sqrt{m} \log^3 n}$
    \item[(iii)] $\mone^f = (1-e^{-d})qm + O \bc{\sqrt{m} \log^3 n}$
    \item[(iv)] $\mone^u = (1-e^{-d})(1-q)m + O \bc{\sqrt{m} \log^3 n}$
\end{itemize}
\end{corollary}

\begin{proof}
Since each test is flipped with probability $p$ and $q$ independently, the claims follow from \Lem~\ref{lem_m0} and the Chernoff bound for the binomial distribution (\Lem~\ref{lem_chernoff_bin}).
\end{proof}

In the following, let $\cE$ be the event that the bounds from \Lem~\ref{lem_m0} and \ref{cor_m0_m1} hold. \gebA{Note that $\cE$ holds with high probability.}

\section{Proof of COMP bound, Theorem \ref{thm_COMP}} \label{sec_COMP}

Recall from \eqref{eq:Nxdef} that we write $\vN_x$ for the number of displayed negative tests that item $x$ appears in (as illustrated by the right branch of Fig.~\ref{neighbor}).
The proof of \Thm~\ref{thm_COMP} is based on two pillars. First, \Lem s~\ref{lem_dist_neg_inf} and \ref{lem_dist_neg_healthy} provide the distribution of $\vN_x$ for healthy and infected individuals, respectively. We will see that these distributions differ according to the infection status of the individual. Second, we will derive a suitable threshold $\alpha \Delta$ via \Lem~\ref{lem_COMP_inf} and \ref{lem_COMP_healthy} to tell healthy and infected individuals apart \whp\ We start by analysing individuals in the infected set $V_1$. Throughout the section, we assume $\alpha \in (q, e^{-d}(1-p) + \bc{1-e^{-d}}q)$.

\begin{lemma} \label{lem_dist_neg_inf}
Given $x \in V_1$, its number of displayed negative tests $\vN_x$ is distributed as $\Bin(\Delta, q)$.
\end{lemma}

\begin{proof}
Any test containing an infected individual is truly positive because of the presence of the infected individual. Since an infected individual is assigned to $\Delta$ different tests and each such test is flipped with probability $q$ independently, the lemma follows immediately.
\end{proof}

Next, we consider the distribution for healthy individuals. Recall that $\cE$ denotes the event that the bounds from \Lem~\ref{lem_m0} and \Cor~\ref{cor_m0_m1} hold.

\begin{lemma} \label{lem_dist_neg_healthy}
Given $x \in V_0$ and \gebA{conditioned on } $\cE$, \change{the total variation distance of the distribution of $\vN_x$ and $\vec{T}_h$ that is distributed as $H \bc{m, m \bc{e^{-d}(1-p)+\bc{1-e^{-d}}q }, \Delta}$  tends to zero with $n$, that is
$$d_{TV}(\vN_x,\vec{T}_h)=n^{-\Omega(1)}$$}
\end{lemma}
\begin{proof}
Since $x$ is healthy, the outcome of all the tests remains the same if it is removed from consideration (if we perform group testing with $n-1$ items and the corresponding reduced matrix).
 
Thus, given $\cE$, we find that with $x$ removed the $\mzero^f, \mzero^u, \mone^f, \mone^u$ still satisfy the bounds from \Cor~\ref{cor_m0_m1}. As a result the number of displayed negative tests (which consist of unflipped truly negative tests and flipped truly positive tests) is given by
\begin{equation} \label{eq:dispneg}
\mzero^u + \mone^f = \left( e^{-d} (1-p) + (1-e^{-d}) q \right) m
+ O \bc{\sqrt{m} \log^3 n}
\end{equation}
Now,  adding $x$ back into consideration: $x\in V_0$ chooses $\Delta$ tests without replacement independently of this. 
\change{Hence, given that the random quantity $\mzero^u + \mone^f = \ell$, the $\vN_x$ (the number of displayed negative tests that item $x$ appears in) is distributed as $H(m, \ell, \Delta)$. Hence, a conditioning argument shows that the linear combination of distribution functions
 $$ \sum_\ell \pr\bc{ \mzero^u + \mone^f = \ell} \pr( H \bc{m,\ell, \Delta} \leq x)$$
  tends to the distribution function of $H \bc{m, m \bc{e^{-d}(1-p)+\bc{1-e^{-d}}q }, \Delta}$ in total variation distance, due to the concentration of $\mzero^u + \mone^f$ as obtained in \Cor~\ref{cor_m0_m1}. 
}
\end{proof}

Moving to the second pillar of the proof, we need to demonstrate that no infected individual is assigned to more than $\alpha \Delta$ displayed negative tests as shown by the following lemma.

\begin{lemma} \label{lem_COMP_inf}
If $c > (1+\eta)\frac{\theta}{1-\theta} \frac{1}{d \KL{\alpha}{q}}$ for some small $\eta>0$, $\vN_x < \alpha \Delta$ for all $x \in V_1$ \whp\
\end{lemma}

\begin{proof}
We have to ensure that $\Pr(\exists x\in V_1:\vN_x \geq \alpha \Delta)=o(1)$.
By \Lem~\ref{lem_dist_neg_inf} and the union bound, we thus need to have
\begin{align*}
  o(1)=  k \cdot \Pr \bc{\vN_x \geq \alpha \Delta : x \in V_1} = k \cdot \Pr \bc{\Bin(\Delta,q) \geq \alpha \Delta} = k
  \cdot \exp \left( - \left(1+ \Delta^{-\Omega(1)} \right) \Delta \KL{\alpha}{q} \right),
\end{align*}
by the Chernoff bound for the binomial distribution (\Lem~\ref{lem_chernoff_bin}). Since $k \sim n^\theta$ and $\Delta = c d (1-\theta) \log n$ \gebA{the following must hold}
\begin{align*}
    \theta-cd(1-\theta)\KL{\alpha}{q}<0
\end{align*}
The lemma follows from rearranging terms \geb{and the fact that if we choose the number of tests slightly above the required number of tests (larger by a factor of $1+\eta$ for $\eta > 0$), the assertion holds \whp\ as $n\rightarrow \infty$}.
\end{proof}

We proceed to show that no healthy individual is assigned to less than $\alpha \Delta$ displayed negative tests.

\begin{lemma} \label{lem_COMP_healthy}
If $c > (1+\eta) \frac{1}{1-\theta} \frac{1}{d \KL{\alpha}{e^{-d}(1-p)+\bc{1-e^{-d}}q}}$ for some small $\eta>0$,
$\vN_x > \alpha \Delta$ for all $x \in V_0$ \whp\
\end{lemma}

\begin{proof}
We need to ensure that $\Pr(\exists x\in V_0: \vN_x < \alpha \Delta)=o(1)$.
Since $\cE$ occurs \whp\  by \Lem~\ref{lem_m0} and \Cor~\ref{cor_m0_m1}, we need to have by \Lem~\ref{lem_dist_neg_healthy} and the union bound that
\begin{align} \label{eq_bound1_healthy}
   (n-k) \cdot \Pr \bc{\vN_x \leq \alpha \Delta | x \in V_0,\cE} \leq n \cdot \Pr \bc{\change{\vec{T}_h} \leq \alpha \Delta} = o(1).
\end{align}
\change{We remind ourselves that $\vec{T}_h\sim H \bc{m, m \bc{e^{-d}(1-p)+\bc{1-e^{-d}}q}, \Delta}$ and }
together with the Chernoff bound for the hypergeometric distribution (\Lem~\ref{lem_chernoff_hyp}) this \gebA{leads to the following condition\footnote{\gebA{Note that the additive rule of the logarithm allows us to move the error term from inside the KL-divergence to outside}}}
\begin{align*}
    1-cd(1-\theta) \KL{\alpha}{(1-p) e^{-d}+(1-e^{-d})q}<0
\end{align*}
in a similar way to the proof of Lemma \ref{lem_COMP_inf}.
The lemma follows from rearranging terms \geb{and the fact that if we choose the number of tests slightly above the required number of tests (larger by a factor of $1+\eta$ for $\eta > 0$), the assertion holds \whp\ as $n\rightarrow \infty$}.
\end{proof}

\begin{proof}[Proof of \Thm~\ref{thm_COMP}]
The theorem is now an immediate consequence of \Lem~\ref{lem_COMP_inf} and \ref{lem_COMP_healthy} which guarantee that \whp\ classifying individuals according to the threshold $\alpha \Delta$ for negative displayed tests recovers $\SIGMA$, and the fact that the choice of $\alpha$ and $d$ is at our disposal.
\end{proof}

\section{Proof of DD bound, Theorem \ref{thm_DD}} \label{sec_DD}

The proof of \Thm~\ref{thm_DD} follows a similar two-step approach as the proof of \Thm~\ref{thm_COMP} by first finding the distribution of $\vP_x$ (the number of displayed positive tests where individual $x$ appears on its own after removing the 
individuals\gebA{, which were declared healthy already, } $V_0 \setminus \zeroplus$, illustrated \gebA{by  DP-S in}  Fig.~\ref{neighbor}). We then threshold the distributions for healthy and infected individuals. To get started, we revise the second bound from \Thm~\ref{thm_COMP} to allow $k n^{-\Omega(1)}$ healthy individuals to not be classified yet after the first step of {\tt DD}. Recall that, we assume $\alpha \in (q, e^{-d}(1-p) + \bc{1-e^{-d}}q)$ and $\beta \in (0, e^{-d}(1-q))$.

\begin{lemma} \label{lem_COMP_healthy_2}
If 
\begin{align*}
	c > (1+\eta)\frac{1}{d \KL{\alpha}{e^{-d}(1-p)+\bc{1-e^{-d}}q}}
\end{align*}
for some small $\eta>0$, we have 
$\change{\abs{\vzeroplus}} = k n^{-\Omega(1)}$ \whp\
\end{lemma}

\begin{proof}
The lemma follows immediately by replacing the r.h.s. of \eqref{eq_bound1_healthy} with $k n^{-\delta}$ for some small $\delta = \delta(\eta)$, rearranging terms and applying Markov's inequality.
\end{proof}

For the next lemmas, we need an auxiliary notation denoting the number of tests $\mzerond$ that only contain individuals from $V_0 \setminus \zeroplus$. In symbols,
\begin{align*}
	\mzerond = \abs{\cbc{a \in F: \partial a \subset V_0 \setminus \zeroplus}}.
\end{align*}

\begin{lemma} \label{lem_m0nd}
If 
\begin{align*}
	c > (1+\eta) \frac{1}{d \KL{\alpha}{e^{-d}(1-p)+\bc{1-e^{-d}}q}}
\end{align*}
for some small $\eta>0$, we have $\mzerond = \bc{1-n^{-\Omega(1)}} e^{-d} m$ with probability $1-o(n^{-2})$.
\end{lemma}

\begin{proof} As in the proof of Lemma \ref{lem_dist_neg_healthy} above, we consider the graph in two rounds: in the first round we consider the tests containing infected individuals. Since each healthy individual $x \in V_0$ does not impact the number of positive and negative tests, we know by \Lem~\ref{lem_m0} that with probability $1-o(n^{-2})$ we find that the number of truly negative tests $\mzero = e^{-d} m + O \bc{\sqrt{m} \log^4 n}$ after the first round.
\change{Furthermore the presence of a healthy individual has no impact on the number of displayed negative tests, as unflipped negative tests remain unflipped and flipped positive tests remain flipped. 
In the second round, we consider the effect of adding healthy individuals into the tests. Knowing the number of negative tests \whp\; we can think of the participation of individuals $x\in V_{0,\text{PD}}$ in these tests as a \textit{balls into bins} experiment. Starting with the  
number of truly negative tests $\mzero$ (given by the first round) we conduct a worst case analysis to see how many of those tests may  include one of the $x\in V_{0,\text{PD}}$. } 
Consider some particular truly negative test $a$. \change{We are interested in the probability that none of the elements of $V_{0,\text{PD}}$ is contained.} The probability that a \change{given individual $x\in V_{0,\text{PD}}$ (knowing that it participates in $N_x \leq \alpha \Delta$ displayed negative tests, which is of lower order than $m$)} is assigned to this test is \change{given by\footnote{\change{We refer the reader to \cite{Gebhard_2021} for two results we use while obtaining \eqref{reform} (apply Claim~7.3 to the binomial coefficients) as well as \eqref{reform2}(apply Claim~7.4 as error corrected version of Bernoulli's inequality).Please note that these bounds in particular hold for $\Delta=\Theta(\log(n))$ and $k\sim n^\theta$.}}
\begin{align}
    \Pr \bc{x \in \partial a\change{|x\in V_{0,\text{PD}}}} &= 1 - \Pr \bc{x \notin \partial a\change{|x\in V_{0,\text{PD}}}}\\
    &= 1 - \sum_{i=0}^{\alpha\Delta}\pr\bc{\vec N_x=i |x \in V_{0,PD}}\binom{m-1}{\Delta-i} \binom{m}{\Delta-i}^{-1} \\
    &\leq 1 - \big(1+n^{-\Omega(1)}\big)\sum_{i=0}^{\alpha\Delta}\pr\bc{\vec N_x=i |x \in V_{0,PD}}\bc{1-\frac{1}{m}}^{\Delta-i}\label{reform}\\
    &\leq  1 -  \big(1+n^{-\Omega(1)}\big)\sum_{i=0}^{\alpha\Delta}\pr\bc{\vec N_x=i |x \in V_{0,PD}}\bc{1-\frac{1}{m}}^{\Delta}= \big(1+n^{-\Omega(1)}\big)\bc{\frac{\Delta}{m}+O(k^{-2})}=\frac{d}{k}+O(k^{-2})\label{reform2}
\end{align}}

\change{We can now calculate the probability that no individual $x \in \zeroplus$ is assigned to $a$, bearing in mind that the size of $\zeroplus$ is random, and that
each such individual is assigned to tests mutually independently. Using \eqref{reform2}, and decomposing the sum into two parts, this is given by (for a given $V$)
\begin{align*}
    \Pr \bc{\big\{\zeroplus \cap \partial a\big\} = \emptyset}&=\sum_{j=0}^{n}\Pr\bc{\abs{\vec{V_{0,\textbf{PD}}}}=j}\Pr\bc{\big\{\zeroplus \cap \partial a \big\}= \emptyset\Big|\abs{\vec{V_{0,\textbf{PD}}}}=j}\\
    &= \sum_{j=0}^{V}\Pr\bc{\abs{\vec{V_{0,\textbf{PD}}}}=j}\bc{1-\frac{d}{k}+O\bc{k^{-2}}}^j
    + \sum_{j=V+1}^{n}\Pr\bc{\abs{\vec{V_{0,\textbf{PD}}}}=j}\bc{1-\frac{d}{k}+O\bc{k^{-2}}}^j \\
    &\geq \sum_{j=0}^{V}\Pr\bc{\abs{\vec{V_{0,\textbf{PD}}}}=j}\bc{1-\frac{d}{k}+O\bc{k^{-2}}}^V
= \pr\bc{ \abs{\vec{V_{0,\textbf{PD}}}} \leq V} \bc{1-\frac{d}{k}+O\bc{k^{-2}}}^V
\end{align*}
By \Lem~\ref{lem_COMP_healthy_2}, we can choose $V = k n^{-\Omega(1)}$ such that $\pr\bc{ \abs{\vec{V_{0,\textbf{PD}}}} \leq V}$ is arbitrarily close to 1, and knowing that $\bc{1-\frac{d}{k}+O\bc{k^{-2}}}^V \simeq \exp(- d V/k) = \exp( -d n^{-\Omega(1)})$
 we find
$$\Pr \bc{\big\{\zeroplus \cap \partial a\big\}=\emptyset}=1-n^{-\Omega(1)}.$$
By combining this with the findings of Lemma~\ref{lem_m0} we find $\Erw \brk{\mzerond} = \bc{1-n^{-\Omega(1)}} e^{-d} m$. The lemma follows by a similar application of the 
\gebA{McDiarmid} inequality as used in the proof of \Lem~\ref{lem_m0}.
}

\end{proof}
\change{Note that, changing the set of tests for a specific individual $x \in V_1 \cup \zeroplus$ shifts $\mzerond$ by at most $\Delta$. Thus, such an individual choosing from this set is not affecting the order of $\mzerond$.}\\
Let $\cF$ be the event that $\mzerond = \bc{1-n^{-\Omega(1)}} e^{-d} m$. By \Lem~\ref{lem_m0nd}, $\Pr \bc{\cF} = 1-o(n^{-2})$ if 
\begin{align*}
	c > (1+\eta) \frac{1}{d \KL{\alpha}{e^{-d}(1-p)+\bc{1-e^{-d}}q}}
\end{align*}
for some small $\eta>0$.
With \Lem~\ref{lem_m0nd} at hand, we are in a position to describe the distribution of $\vP_x$ for healthy and infected individuals \change{(recall the definition of $\vP_x$ in \eqref{eq:Pxdef})}. Let us start with infected individuals.


\begin{lemma} \label{lem_dist_pos_inf}
Given $x \in V_1$ and conditioned on $\cF$, \change{the total variation distance between
$\vP_x$ and $\vec{Q}_H$, a random variable with hypergeometric distribution
$H\bc{m, m e^{-d} (1-q), \Delta}$, tends to zero with $n$, that
is 
$$d_{TV} \left( \vP_x, \vec{Q}_H \right) = n^{-\Omega{(1)}}.$$}
\end{lemma}
\begin{proof}
\change{
We are interested in the neighborhood structure of one given infected individual $x \in V_1$, and 
we check how the remaining individuals influence the  test types. 
In particular we are interested in the number of tests $a \in F$ such that $\partial a \subset V_0 \setminus \zeroplus$ are contained in the neighborhood of an infected individual $x$. Knowing the total number of tests $m$ and fixed degree $\Delta$, for a given value of the random quantity $\mzerond =\ell$, we find that this quantity of interest follows a $H \bc{m,\ell, \Delta}$-distribution. 
Given $\cF$, \Lem~\ref{lem_m0nd} gives that $\mzerond$ is highly concentrated,
 $$\mzerond = \bc{1-n^{-\Omega(1)}} e^{-d} m$$
 with high probability. Hence a conditioning argument, similar to Lemma B.2, shows that the 
 linear combination of distribution functions
 $$ \sum_\ell \pr( \mzerond = \ell) \pr( H \bc{m,\ell, \Delta} \leq x)$$
  tends to the distribution function of $H \bc{m, m e^{-d}, \Delta}$ in total variation distance, due to the concentration result obtained in Lemma~\ref{lem_m0nd}.
%
Since each test featuring $x$ will truly be positive (as we assume $x$ to be infected) and will be displayed positive with probability $1-q$ independently, the lemma follows immediately.}
%
\end{proof}

To describe the distribution of $\vP_x$ for healthy individuals, let us introduce the random variable $\vP_x(P)$, which is $\vP_x$ conditioned on the individual appearing in $P$ displayed positive tests, as follows:
\begin{align*}
	\Pr \bc{\vP_x(P) = t} = \Pr \bc{\vP_x = t | \vN_x = \Delta - P} 
\end{align*}
Then, we find for healthy individuals the following conditional distribution.

\begin{lemma} \label{lem_dist_pos_healthy}
Given $x \in V_0$ \change{,conditioned on $\cE$} and $\cF$, \change{the total variation distance between $\vP_x(P)$ and \\
$\vec{B}_h\sim H\bc{m \bc{e^{-d}p+(1-e^{-d})(1-q)},m \bc{e^{-d}p}, P}$ tends to zero with $n$. That is
$$d_{TV}(\vP_x(P),\vec{B}_h)=n^{-\Omega(1)}.$$}
\end{lemma}
\begin{proof}
We proceed with the same exposition \change{and reasoning }as in the proof of \Lem~\ref{lem_dist_pos_inf}. 
\change{Due to the fact that $x$ is healthy we can remove it without affecting the test result. Therefore we can analyse its neighborhood structure induced by the pooling graph while excluding it. Since by assumption individual $x \in V_0$ is assigned to exactly $P$ displayed positive and the total number of displayed positive test is given by $\mzero^f+\mone^u$, we see that $\vP_x(P)$ is $H\bc{\mzero^f+\mone^u, \mzerond, P}$-distributed.
Due to the fact that the event $\cE$  pinpoints the amount of displayed positive and negative tests we can derive the distribution of neighbors the individual may choose from. Recalling the results of \Cor~\ref{cor_m0_m1}, we see that \whp
\begin{align*}
    \mzero^f &= e^{-d}pm + O \bc{\sqrt{m} \log^3 n},\\
    \text{and }\mone^u &= (1-e^{-d})(1-q)m + O \bc{\sqrt{m} \log^3 n}.
\end{align*}
Furthermore we get from \Lem~\ref{lem_m0nd} that \whp
$$\mzerond = \bc{1-n^{-\Omega(1)}} e^{-d} m.$$
Now we apply the concentration results obtained in \Cor~\ref{cor_m0_m1} and \Lem~\ref{lem_m0nd} to obtain a linear combination of distribution functions
 $$ \sum_{\ell,v} \pr( \mzerond = \ell,\mzero^f+\mone^u=v) \cdot \pr( H \bc{v,\ell, \Delta} \leq x)$$
 that tends to $H\bc{m \bc{e^{-d}p+(1-e^{-d})(1-q)},m e^{-d}, P}$.
 The lemma follows since truly negative tests get flipped independently with probability $p$. 
 }
\end{proof}

Having derived the distributions for $\vP_x$ for $x \in V_1$ and $\vP_x(P)$ for $x \in V_0$ we can now determine a threshold $\beta \Delta$ of displayed positive tests where the individual appears only with individuals from the set $V_0 \setminus \zeroplus$ such that we can tell $V_1$ and $\zeroplus$ apart and thus recover $\SIGMA$. Let us start with infected individuals.

\begin{lemma} \label{lem_DD_inf}
 As long as
\begin{align*}
    c > (1+\eta) \max \cbc{\frac{1}{d \KL{\alpha}{e^{-d}(1-p)+\bc{1-e^{-d}}q}}, \frac{\theta}{1-\theta} \frac{1}{d \KL{\beta}{(1-q)e^{-d}}}}
\end{align*}
for some small $\eta>0$, 
we have $\vP_x > \beta \Delta$ for all $x \in V_1$ \whp\
\end{lemma}

\begin{proof}
We need to ensure that $\Pr(\exists x\in V_1:\vP_x < \beta \Delta)=o(1)$.
For the bound on $c$ from the lemma, we know that $\cF$ occurs \whp\ by \Lem~\ref{lem_m0nd}.
In combination with \Lem~\ref{lem_dist_pos_inf} and the union bound we need to ensure that
\begin{align} \label{DD11}
    k \cdot \Pr \bc{\vP_x \leq \beta \Delta|x\in V_1,\cF} = k \cdot 
 \pr( \change{\vec{Q}_H}  \leq \beta \Delta) + k n^{-\Omega(1)} = o(1),
\end{align}
\change{where as before $\vec{Q}_H$ is a random variable with hypergeometric distribution}
$H\bc{m, m e^{-d} (1-q), \Delta}$.
Using the Chernoff bound for the hypergeometric distribution (\Lem~\ref{lem_chernoff_hyp}), \gebA{the following condition for \eqref{DD11} to hold arises}
\begin{align}\label{DD12}
    \theta-cd(1-\theta)\KL{\beta}{(1-q)e^{-d}}<0
\end{align}

The lemma follows from rearranging terms in \eqref{DD12} \geb{and the fact that if we choose the number of tests slightly above the required number of tests (larger by a factor of $1+\eta$ for $\eta > 0$), the assertion holds \whp\ as $n\rightarrow \infty$}.
\end{proof}

We proceed with the set of individuals $\zeroplus$.

\begin{lemma} \label{lem_DD_healthy}
As long as 
\begin{align*}
    c &> (1+\eta) \max \Bigg\{\frac{1}{d \KL{\alpha}{e^{-d}(1-p)+\bc{1-e^{-d}}q}}, \\ 
    & \qquad \max_{1-\alpha\leq z \leq 1} \cbc{ \frac{1}{1-\theta} \frac{1}{d \bc{ \KL{z}{e^{-d}p + (1-e^{-d}) (1-q)} + z \KL{\frac \beta z}{\frac{e^{-d}p}{e^{-d}p + (1-e^{-d}) (1-q)}}}}}\Bigg\}
\end{align*}
for some small $\eta>0$, 
we have $\vP_x < \beta \Delta$ for all $x \in \zeroplus$ \whp\
\end{lemma}

\begin{proof}
We need to ensure that $\Pr(\exists x\in V_{0,PD}:\vP_x > \beta \Delta)=o(1)$.
For the bound on $c$ from the lemma, we know that $\cF$ occurs \whp\ by \Lem~\ref{lem_m0nd}. Moreover, $\cE$ occurs \whp\ by \Lem~\ref{lem_m0} and \Cor~\ref{cor_m0_m1}. We write $w= e^{-d} p + \bc{1-e^{-d}(1-q)}$ for brevity.
Combining this fact with \Lem~\ref{lem_dist_neg_healthy} and \ref{lem_dist_pos_healthy} we need to ensure
\begin{align} \label{eq_DD2}
    &(n-k) \sum_{P=(1-\alpha) \Delta}^{\Delta} \Pr \bc{\vN_x = \Delta - P | x \in V_0,\cE} \Pr \bc{\vP_x(P) \geq \beta \Delta | x \in V_0,\cF} \\
    &= \bc{1-n^{-\Omega(1)}} n \sum_{P=(1-\alpha)\Delta}^\Delta \Pr \bc{\change{\vec{T_h}} = P} \cdot \Pr \bc{\change{\vec{B_h}} \geq \beta \Delta} = o(1) \label{eq_c4_exp}
\end{align}
\change{We remind ourselves that
\begin{align*}
    \vec{T_h} &\sim H \bc{m, m \bc{e^{-d}(1-p)+\bc{1-e^{-d}}q }, \Delta}\\
    \text{and }\quad \vec{B_h}& \sim H\bc{m \bc{e^{-d}p+(1-e^{-d})(1-q)},m \bc{e^{-d}p}, P}.
\end{align*}
Now} by the Chernoff bound for the hypergeometric distribution (\Lem~\ref{lem_chernoff_hyp}) and setting $z=P/\Delta$, \geb{we establish the following two bounds for the probability terms:
\begin{align}
\Pr \bc{H \bc{m, m \bc{w+ n^{-\Omega(1)}}, \Delta} = P}=\exp\bc{-(1+n^{-\Omega(1)})\Delta \bc{\KL{z}{w
    }}}
\end{align}
\begin{align}
&\Pr \bc{H\bc{m \bc{w+n^{-\Omega(1)}},m \bc{e^{-d}p+n^{-\Omega(1)}}, P} \geq \beta \Delta}\notag\\&= \exp\bc{-\bc{1+n^{-\Omega}}z\Delta\vecone \cbc{\beta > \frac{ze^{-d}p}{w}} z \KL{\frac \beta z}{\frac{e^{-d}p}{w
    }}}\label{chernoff_apply}
\end{align}}
(Note that the indicator in \eqref{chernoff_apply} appears due to the condition given by \Lem~\ref{lem_chernoff_hyp})
We reformulate the left-hand-side of \eqref{eq_c4_exp} to
\begin{align*}
    &n \sum_{P=(1-\alpha)\Delta}^\Delta \exp \bc{-(1+o(1))\Delta \bc{\KL{z}{w
    }+\vecone \cbc{\beta > \frac{ze^{-d}p}{w}} z \KL{\frac \beta z}{\frac{e^{-d}p}{w
    }}}} \\
    &\qquad \qquad = \bc{1+n^{-\Omega(1)}} n \max_{1-\alpha \leq z \leq 1} \Bigg\{ \exp \Bigg(-(1+o(1))\Delta \Bigg(\KL{z}{w 
    } 
    +\vecone \cbc{\beta > \frac{ze^{-d}p}{w}} z \KL{\frac \beta z}{\frac{e^{-d}p}{w 
    }}\Bigg)\Bigg)\Bigg\}
\end{align*}
where the second equality follows since the sum consists of $\Theta(\Delta) = \Theta(\log n)$ many summands. Since $\Pr \bc{\cF} = 1 - n^{-\Omega(1)}$ for our choice of $c$ by \Lem~\ref{lem_m0nd} rearranging terms readily yields that the expression in \eqref{eq_DD2} is indeed of order $o(1)$.\\\\ \geb{To see this, we remind ourselves that by definition $\Delta=c d \log\bc{\frac{n}{k}}=(1-\theta)c d \log(n)$. Furthermore we plug in the definition for $w= e^{-d} p + \bc{1-e^{-d}(1-q)}$. In the end we have to ensure that
$$1<(1-\theta)c d\Bigg(\KL{z}{w 
    } 
    +\vecone \cbc{\beta > \frac{ze^{-d}p}{e^{-d} p + \bc{1-e^{-d}(1-q)}}} z \KL{\frac \beta z}{\frac{e^{-d}p}{e^{-d} p + \bc{1-e^{-d}(1-q)} 
    }}\Bigg)$$
    }\gebA{We solve this inequality for $c$. As we are only interested in a worst case bound, the assertion follows from the non-negativity of $\KL{*}{*}$.}

\end{proof}

\begin{proof}[Proof of \Thm~\ref{thm_DD}]
The theorem is now immediate from \Lem~\ref{lem_COMP_inf}, \ref{lem_COMP_healthy_2}, \ref{lem_DD_inf} and \ref{lem_DD_healthy} and the fact that the choice of $\alpha, \beta$ and $d$ is at our disposal.
\end{proof}

\section{Comparison of the noisy {\tt DD} and {\tt COMP} bounds} \label{sec_Com_DD_COMP}
The following section is intented to provide sufficient conditions under which the {\tt DD} algorithm attains reliable performance requiring fewer tests than the {\tt COMP}.
However, these conditions are not necessary and {\tt DD} might (and for all performed simulations does) \gebA{require fewer tests than} {\tt COMP} for even wider settings.

\begin{proof}[Proof of \Prop~\ref{prop_comp_dd}]
In order to prove the proposition, we need to find conditions under which
\begin{align*}
    \min_{\alpha, d} \max \cbc{b_1(\alpha, d), b_2(\alpha, d)} \geq \min_{\alpha, \beta, d} \max \cbc{c_1(\alpha, d), c_2(\alpha, d), c_3(\beta, d), c_4(\alpha, \beta, d)}
\end{align*}
We write $\alpha^*$ and $d^*$ for the values that minimise the maximum of the two terms at the LHS, at which point we know that
$b_1(\alpha^*, d^*) = b_2(\alpha^*, d^*)$.
Then it is sufficient to show that there exists $\beta^*$ such that
$$
b_1(\alpha^*, d^*)= b_2(\alpha^*, d^*) \geq
\max \cbc{c_1(\alpha^*, d^*), c_2(\alpha^*, d^*), c_3(\beta^*, d^*), c_4(\alpha^*, \beta^*, d^*)}
$$
By inspection for any $\alpha$ and $d$
$b_1(\alpha, d)=c_1(\alpha, d)$ and $b_2(\alpha, d) \geq c_2(\alpha, d)$ since $\theta \in (0,1)$.

Next, we will show that $b_2(\alpha, d) \geq c_4(\alpha, \beta, d)$ for any $\alpha, \beta$ in the respective bounds and $d \in (0,\infty)$.  Writing $w = e^{-d} p + (1-e^{-d}) (1-q)$, and recalling that by assumption that $\alpha \leq 1-w$ (or $w \leq 1-\alpha$) we readily find that
\begin{equation} \label{eq:b2c4}
  \KL{\alpha}{1-w} = \min_{1-\alpha \leq z \leq 1} \bc{ \KL{z}{w}} \leq
\min_{1-\alpha \leq z \leq 1} \bc{ \KL{z}{w} + z \vecone \cbc{\beta > \frac{z e^{-d}p}{w}} \KL{\frac \beta z}{\frac{e^{-d}p}{w}}}
\end{equation}
where the first equality follows since $\KL{\alpha}{1-w}=\KL{1-\alpha}{w}$ and $\KL{z}{w} > \KL{1-a}{w}$ for any $z > 1-\alpha$. The bound follows. Note that \eqref{eq:b2c4} indeed holds for any choice of $\alpha, \beta$ and $d$ in the respective bounds stated in the theorem.

Finally, we need to demonstrate that $c_3 (\beta^*, d^*) \leq b_2 (\alpha^*, d^*)$. Since $\beta$ is not an optimisation parameter in $b_2(\alpha^*, d^*)$ and the bound in \eqref{eq:b2c4} holds for any value of $\beta$, we can simply set it to the value that minimizes $c_3(\beta^*, d^*)$ which is $\beta=1/\Delta$ and for which we find
\begin{align*}
    c_3(\beta^*, d^*) = \frac{\theta}{1-\theta} \frac{1}{d^* \log \bc{1-e^{-d^*}(1-q}}.
\end{align*}
Thus, to obtain the desired inequality we need to ensure that for the optimal choice $\alpha^*$ from {\tt COMP}
\begin{align*}
    \theta \KL{\alpha^*}{e^{-d^*}(1-p) + \bc{1-e^{-d^*}}q} &\leq -\log \bc{1-e^{-d^*}(1-q)}
\end{align*}

Using the bound
\begin{align*}
    \theta \KL{\alpha}{e^{-d}(1-p) + \bc{1-e^{-d}}q} &\leq - \theta \log \bc{1- \left( e^{-d}(1-p) + \bc{1-e^{-d}}q \right)} \\ &\leq -\log \bc{1-\left( e^{-d}(1-p) + \bc{1-e^{-d}}q \right)}
\end{align*}
which is obtained by setting $\alpha = 1/\Delta$, we find that $c_3 (\beta^*, d^*) \leq b_2 (\alpha^*, d^*)$ if
\begin{align*}
    -\log \bc{1-e^{-d^*}(1-q)} \geq -\log \bc{1-e^{-d^*}(1-p) + \bc{1-e^{-d^*}}q} \Leftrightarrow
    e^{-d^*}p \geq q
\end{align*}

\end{proof}

As mentioned before, due to bounding $b_2(\alpha^*, d^*)$ the result is not sharp. However, one immediate consequence of \Prop~\ref{prop_comp_dd} is that {\tt DD} is guaranteed to \gebA{require fewer tests than} {\tt COMP} for the reverse Z channel.

\section{Relation to Bernoulli testing}\label{sec_Rel_Ber}

In the noiseless case \cite{Johnson_2019} shows that the constant column weight design (where each individual joins exactly $\Delta$ different tests)  requires fewer tests to recover $\SIGMA$ than the i.i.d. (Bernoulli pooling) design (where each individual is included in each test with a certain probability independently). In this section we show that in the noisy case, the {\tt COMP} algorithm requires fewer tests for the constant column weight design than for the i.i.d.  design, and derive sufficient conditions under which the same is true for the noisy {\tt DD} algorithm.

To get started, let us state the relevant bounds for the Bernoulli design, taken from \cite[Theorem 5]{Johnson_2018} and rephrased in our notation.




\begin{proposition}[Noisy {\tt COMP} under Bernoulli] \label{prop_COMP_Bernoulli}
Let $p,q\geq 0$, $p+q<1$, $d \in (0,\infty)$, $\alpha \in (q, e^{-d}(1-p) + \bc{1-e^{-d}}q)$. Suppose that $0<\theta<1$ and $\eps>0$ and let
\begin{align*}
    \mcomp^{\text{Ber}} &= \mcomp^{\text{Ber}}(n,\theta, p, q) = \min_{\alpha, d} \max \cbc{b_1(\alpha, d), b_2(\alpha, d)} k \log(n/k) \\
&\text{where} \qquad b_1(\alpha, d) = \frac{\theta}{1-\theta} \frac{1}{k\KL{\alpha d/k}{qd/k}} \\
&\text{and} \qquad b_2(\alpha, d) = \frac{1}{1-\theta}\frac{1}{k\KL{\alpha d/k}{(e^{-d}(1-p)+(1-e^{-d})q)d/k}}
\end{align*}
If $m>(1+\eps) \mcomp^{\text{Ber}}$, {\tt COMP} will recover $\SIGMA$ under the Bernoulli test design \whp\ given $\vec G, \hat \SIGMA$.
\end{proposition}



\begin{proposition}[Noisy {\tt DD} under Bernoulli] \label{prop_DD_Bernoulli}
Let $p,q\geq 0$, $p+q<1$, $d \in (0, \infty)$, $ \alpha \in (q, e^{-d}(1-p) + \bc{1-e^{-d}}q)$ and $\beta \in (e^{-d}p, e^{-d}(1-q))$. Suppose that $0<\theta<1, \zeta \in (0,\theta)$ and $\eps>0$ and let
\begin{align*}
    \mdd^{\text{Ber}} &= \mdd^{\text{Ber}}(n,\theta, p, q) = \min_{\alpha, \beta, d} \max \cbc{c_1(\alpha, d), c_2(\alpha, d), c_3(\beta, d), c_4(\beta, d)} k \log(n/k) \\
&\text{where} \qquad c_1(\alpha, d) = \frac{\theta}{1-\theta} \frac{1}{k\KL{\alpha d/k}{qd/k}} \\
&\text{and} \qquad c_2(\alpha, d) = \frac{1-\zeta}{1-\theta}\frac{1}{k\KL{\alpha d/k}{(e^{-d}(1-p)+(1-e^{-d})q)d/k}} \\
&\text{and} \qquad c_3(\beta, d) = 
\frac{\theta}{1-\theta} \frac{1}{k\cdot \KL{\beta d/k}{e^{-d}(1-q)d/k}} \\
&\text{and} \qquad c_4(\beta, d) = \frac{\zeta}{1-\theta} \frac{1}{k\cdot \KL{\beta d/k}{e^{-d}p d/k}}
\frac{}{}
\end{align*}
If $m>(1+\eps) \mdd^{\text{Ber}}$, {\tt DD} will recover $\SIGMA$ under the Bernoulli test design \whp\ given $\vec G, \hat \SIGMA$.
\end{proposition}


To compare the bounds of the Bernoulli and constant-column test design we employ the following handy observation.

\begin{lemma}\label{DKL_reform}
Let $0 < x, y < 1$ and $d>0$ be constants independent of $k$. As $k \to \infty$
\begin{align*}
    k \KL{\frac{xd}{k}}{\frac{yd}{k}}=d \left(\KL{x}{y}+v(x,y)\right) +o(1/k)
\end{align*}
with
\begin{align}\label{calc_imp}
v(x,y)=y-x+(1-x)\log\left(\frac{1-y}{1-x}\right) \leq 0
\end{align}
\end{lemma}

\begin{proof}
Applying the definition of the Kullback-Leibler divergence and Taylor expanding the logarithm we obtain
\begin{align*}
    k \cdot \KL{\frac{xd}{k}}{\frac{yd}{k}}=& 
    xd\cdot \log\left(\frac{x}{y}\right)+(k-xd)\left(\log\left(1-\frac{xd}{k}\right)-\log\left(1-\frac{yd}{k}\right)\right)\\
    &= xd\cdot \log\left(\frac{x}{y}\right)+(k-xd)\left(-\frac{xd}{k}+\frac{yd}{k}+o\left(\frac{1}{k^{2}}\right)\right)\\
    &=d\left(x\cdot \log\left(\frac{x}{y}\right)-x+y\right)+o(1/k)\\
    &=d\left(\KL{x}{y}+y-x-(1-x)\log\left(\frac{1-x}{1-y}\right)\right)+o(1/k).
\end{align*}
We can bound $v(x,y)$ from above by writing the final term as
$(1-x) \log \left( 1 + \frac{x-y}{1-x} \right) \leq (1-x) \frac{x-y}{1-x} = x-y$, using the standard linearisation of the logarithm.
\end{proof}

We are now in a position to prove \Prop~\ref{prop_Ber_COMP} and \ref{prop_Ber_DD}.

\begin{proof}[Proof of \Prop~\ref{prop_Ber_COMP}]
The lemma follows by comparing the bounds from \Thm~\ref{thm_COMP} and \Prop~\ref{prop_COMP_Bernoulli} and applying \Lem~\ref{DKL_reform}.
\end{proof}

\begin{proof}[Proof of \Prop~\ref{prop_Ber_DD}]
As evident from \Cor~\ref{cor_Z_DD}, the fourth bound $c_4(\alpha, \beta, d)$ vanishes under the Z channel. Now comparing the bounds from \Thm~\ref{thm_DD} and \Prop~\ref{prop_DD_Bernoulli}, observing that $(1-\zeta)/(1-\theta) > 1$ for $\zeta < \theta$ and applying \Lem~\ref{DKL_reform}  immediately implies the lemma.
\end{proof}

\section{Notes on Corollary~\ref{thm:shanCAP}}\label{Notes_on_Capacity}
\geb{
\begin{lemma} \label{lem:shanCAP} If $p + q < 1$
the Shannon capacity of the $p-q$ channel of Figure~\ref{pqnoise} measured in nats is
\begin{equation} \label{eq:chancap}
\Cchan = \KL{q}{\frac{1}{1+e^\phi}}
= \KL {p}{\frac{1}{1+e^{-\phi}}},
\end{equation}
where $\phi = (h(p) -h(q))/(1-p-q)$.
This is achieved by taking 
\begin{equation} \label{eq:sigstrat}
\pr(X=0) = \frac{1}{1-p-q}
\left( \frac{1}{1+e^\phi} - q \right).\end{equation}
\end{lemma}
\gebA{Please note that the proof might be a standard result for readers from some research communities, but for others it might be less standard. Therefore we state it here to prevent the interested (but unfamiliar) reader from a long textbook search. }
\begin{proof}  Write $\pr(X=0) = \gamma$ and $\pr(Y= 0) = T(\gamma) := (1-p) \gamma + q (1-\gamma)$. Then since the mutual information
\begin{equation} \label{eq:mutinf}
I(X;Y) = h(Y) - h(Y | X) = h \left(T(\gamma) \right) - \left( \gamma h(p) + (1-\gamma) h(q) \right),\end{equation}
we can find the optimal $T$ by solving 
$$ 0 = \frac{\partial}{\partial \gamma} I(X;Y) =
(1-p-q)  \log \left( \frac{1-T(\gamma)}{T(\gamma)} \right) - \left( h(p) - h(q) \right),$$
which implies that the optimal $T^* = 1/(1+e^\phi)$. We can solve for this for $\gamma^* = (T^*-q)/(1-p-q)$ to find the expression above. As $\frac{\partial}{\partial^2 \gamma} I(X;Y)<0$ it is indeed a maximum. Substituting this in \eqref{eq:mutinf} we obtain that the capacity is given by
\begin{eqnarray}
h(T^*) - \left( \gamma^* h(p) + (1-\gamma^*) h(q) \right)
& = & h \left( \frac{1}{1+e^\phi} \right) - \left( (T^*-q) \phi +
h(q) \right) \nonumber \\
& = & \log(1+e^\phi) - \phi(1-q) - h(q) \label{eq:tocompare} \\
& = & \KL {q}{1/(1+e^\phi)} \nonumber
\end{eqnarray}
as claimed in the first expression in \eqref{eq:chancap} above. We can see that the second expression in \eqref{eq:chancap} matches the first by writing the corresponding expression as $\KL{1-p}{1/(1+e^\phi)} = \log(1+e^\phi) - 
\phi p - h(p)$, which is equal to \eqref{eq:tocompare} by the definition of $\phi$.
\end{proof}
Note that this result suggests a choice of density for the matrix: since each test is negative with probability $e^{-d}$, equating this with \eqref{eq:sigstrat} suggests that we take
$$ d = \dch = \log(1-p-q) - \log \left( \frac{1}{1+e^\phi} - q \right).$$
This is unlikely to be optimal in a group testing sense, since we make different inferences from positive and negative tests, but gives a closed form expression that may perform well in practice. For the noiseless and BSC case observe that $\phi= 0$, and we obtain $\dch = \log 2$.}
\section{Illustration of bounds for Z, reverse Z channel and the BSC}\label{Appendix_Illustration}

\begin{figure}[H]
    \centering
    \begin{subfigure}{.4\textwidth}
    \centering
        \includegraphics[width=1\linewidth]{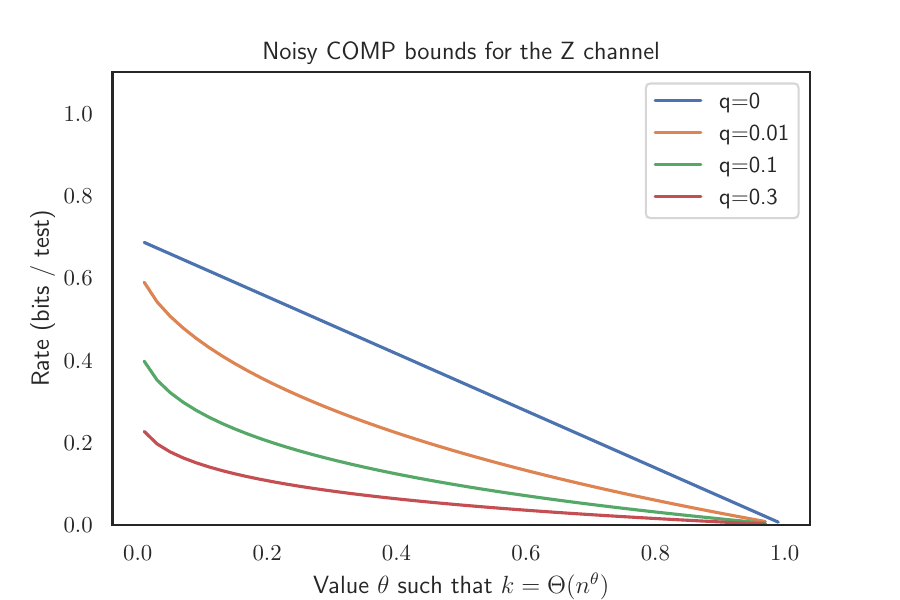}

    \end{subfigure}%
    \begin{subfigure}{.4\textwidth}
    \centering
        \includegraphics[width=1\linewidth]{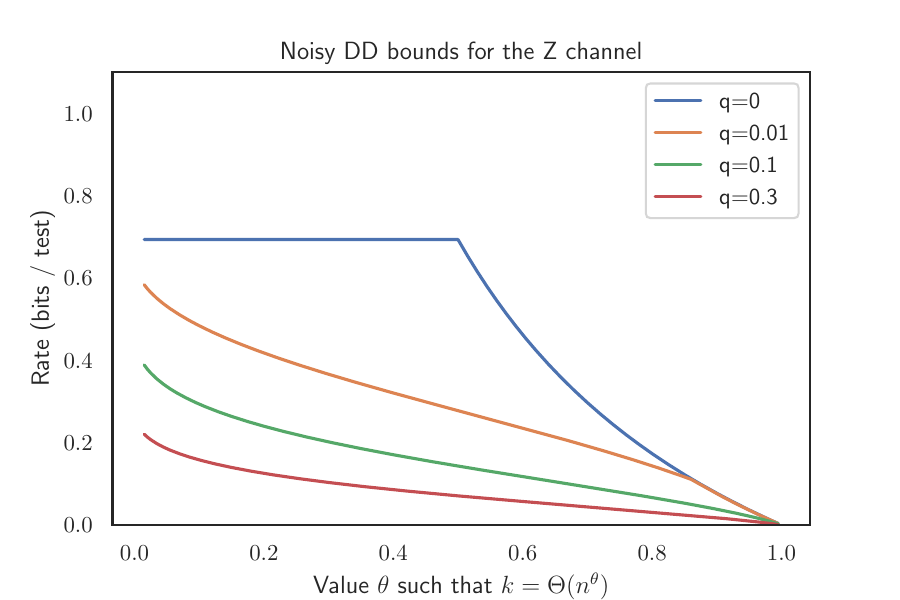}

    \end{subfigure}
    \caption{Illustration of achievability bounds for noisy {\tt COMP} and {\tt DD} under the Z channel. \geb{ (Note for black and white prints: The solid lines as well as the dashed lines in the diagram are in the same order as given in the legend from top to bottom) }} 
    \label{fig_Z_channel}
\end{figure}

\begin{figure}[H]
    \centering
    \begin{subfigure}{.4\textwidth}
    \centering
        \includegraphics[width=1\linewidth]{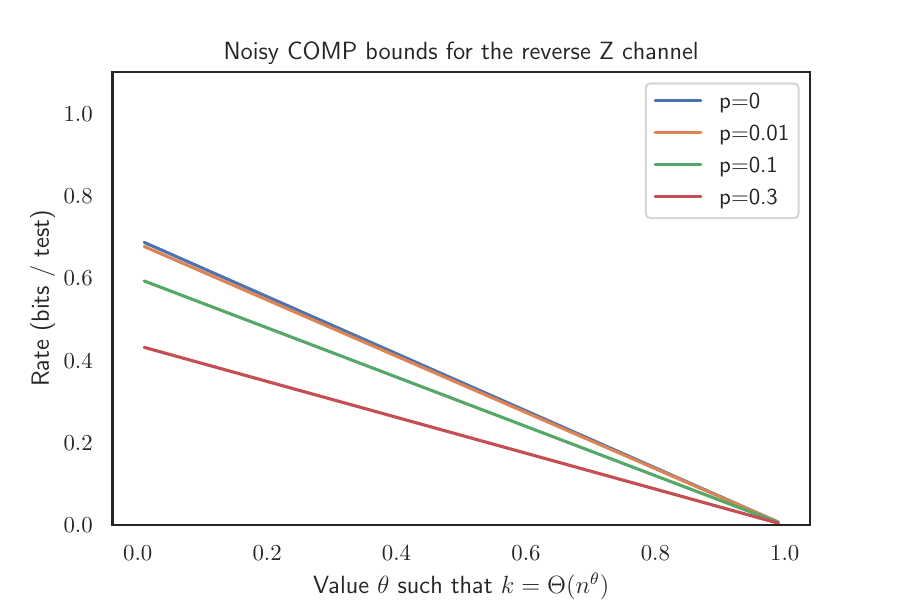}

    \end{subfigure}%
    \begin{subfigure}{.4\textwidth}
    \centering
        \includegraphics[width=1\linewidth]{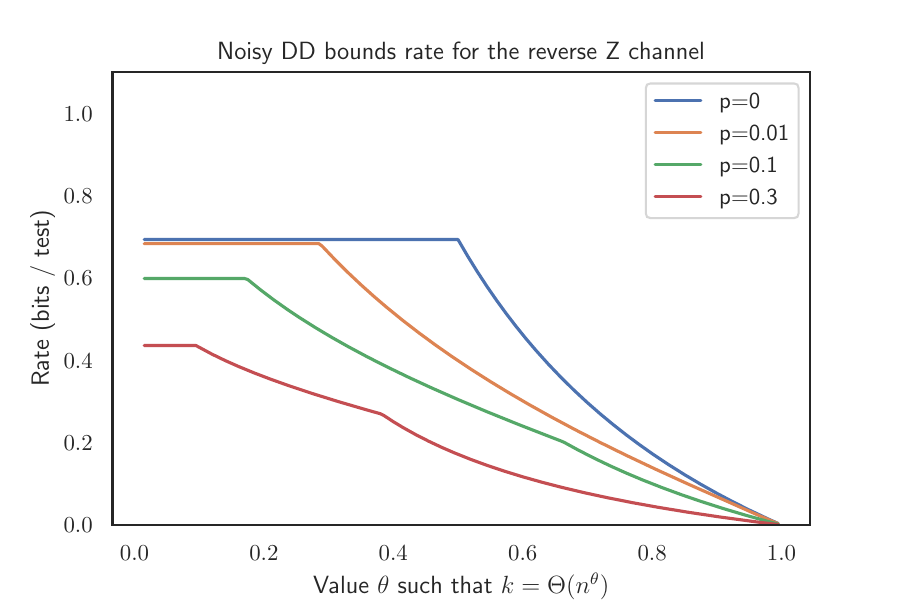}

    \end{subfigure}
    \caption{Illustration of achievability bounds for noisy {\tt COMP} and {\tt DD} under the reverse Z channel. \geb{ (Note for black and white prints: The solid lines as well as the dashed lines in the diagram are in the same order as given in the legend from top to bottom) }}
    \label{fig_rev_Z_channel}
\end{figure}


\begin{figure}[H]
    \centering
    \begin{subfigure}{.4\textwidth}
    \centering
        \includegraphics[width=1\linewidth]{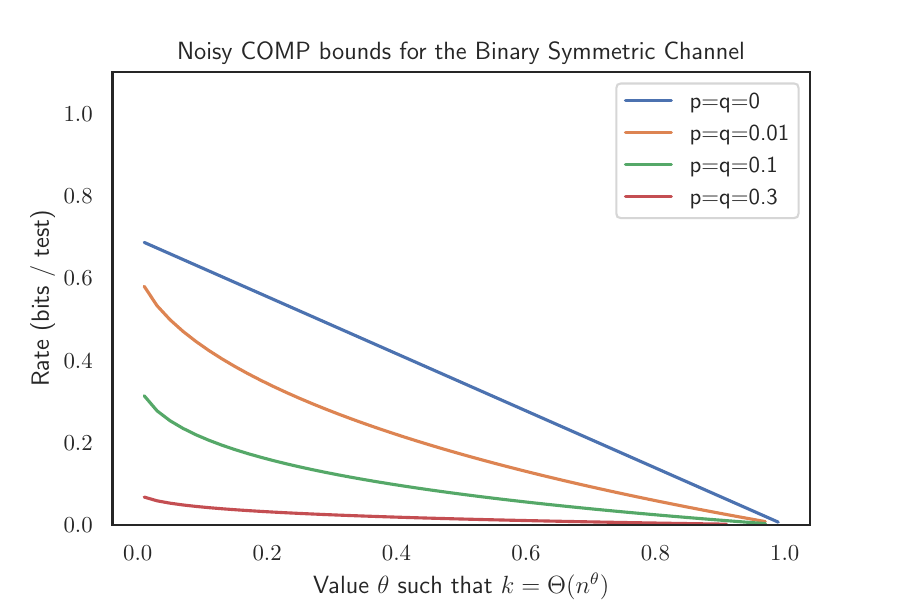}

    \end{subfigure}%
    \begin{subfigure}{.4\textwidth}
    \centering
        \includegraphics[width=1\linewidth]{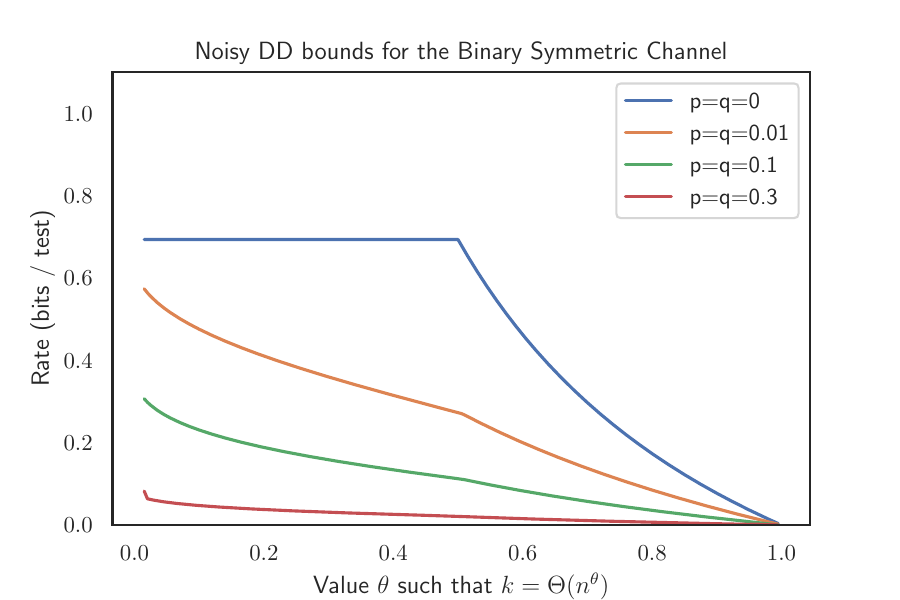}

    \end{subfigure}
    \caption{Illustration of achievability bounds for noisy {\tt COMP} and {\tt DD} under the Binary Symmetric Channel. \geb{ (Note for black and white prints: The solid lines as well as the dashed lines in the diagram are in the same order as given in the legend from top to bottom) }}
    \label{fig_BS_channel}
\end{figure}

\begin{figure}[H]
    \centering
    \includegraphics[width=0.6\linewidth]{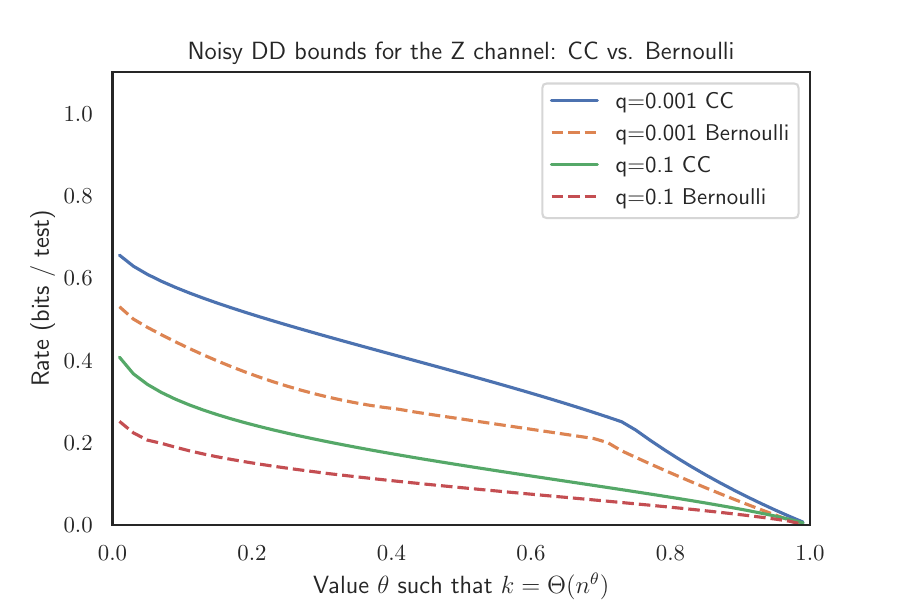}
\caption{Comparison of the noisy {\tt DD} rates under Bernoulli pooling (\cite{Johnson_2018}) with the {\tt DD} bounds with constant-column design as provided in the paper at hand within the Z-Channel.\geb{ (Note for black and white prints: The solid lines as well as the dashed lines in the diagram are in the same order as given in the legend from top to bottom). }}
\label{fig_Z_BervsCC}
\end{figure}

\section*{Acknowledgment}

The authors would like to thank two anonymous referees for their detailed reading of this paper and for the suggestions they made to improve its presentation. Oliver Gebhard and Philipp Loick are supported by DFG CO 646/3.


\end{document}